\documentclass[11pt,reqno]{article}
\usepackage{graphicx,amsmath,amssymb,setspace}
\usepackage{natbib}

\newtheorem{thm}{Theorem}[section]
\newtheorem{lemma}{Lemma}[section]
\newtheorem{remark}{Remark}[section]
\newtheorem{cor}{Corollary}[section]

\newtheorem{proposition}{Proposition}[section]
\newenvironment{proof}{\textsc{Proof:}}{\mbox{ } \hfill $\Box$ \vspace{2mm}}
\numberwithin{equation}{section}

\renewcommand{\P}{\mathbb{P}}

\newcommand{\ls}{\lambda^S}
\newcommand{\lo}{\lambda^{O}}

\renewcommand{\L}{\mathcal{L}}

\newcommand{\dd}{\displaystyle}

\begin{document}

\title{{\LARGE \textbf{\textsc{Optimal Reversible Annuities to Minimize the
Probability of Lifetime Ruin}}}}

\author{Ting Wang 
\thanks{Department of Mathematics, University of Michigan, Ann Arbor, MI 48109,  email: wting@umich.edu.}
\and Virginia R. Young \thanks{ Department of Mathematics,
University of Michigan, Ann Arbor, Michigan, 48109, email:vryoung@umich.edu.  V.\ R.\ Young thanks the Nesbitt Professorship of Actuarial Mathematics for financial support.} }

\date{11 June 2009}

\maketitle

\begin{abstract}
We find the minimum probability of lifetime ruin of an investor who can invest in a market with a risky and a riskless asset and who can purchase a reversible life annuity.  The surrender charge of a life annuity is a proportion of its value.  Ruin occurs when the total of the value of the risky and riskless assets and the surrender value of the life annuity reaches zero.  We find the optimal investment strategy and optimal annuity purchase and surrender strategies in two situations:  (i) the value of the risky and riskless assets is allowed to be negative, with the imputed surrender value of the life annuity keeping the total positive; or (ii) the value of the risky and riskless assets is required to be non-negative.  In the first case, although the individual has the flexiblity to buy or sell at any time, we find that the individual will not buy a life annuity unless she can cover all her consumption via the annuity and she will never sell her annuity.  In the second case, the individual surrenders just enough annuity income to keep her total assets positive.  However, in this second case, the individual's annuity purchasing strategy depends on the size of the proportional surrender charge.  When the charge is large enough, the individual will not buy a life annuity unless she can cover all her consumption, the so-called safe level.  When the charge is small enough, the individual will buy a life annuity at a wealth lower than this safe level.

\end{abstract}

\section{Introduction}

The so-called ``annuity puzzle'' is that in financial markets for which annuity purchase is not mandatory, the volume of voluntary purchases by retirees is much smaller than predicted by models, such as those proposed by \cite{Yaari},  \cite{Richard} and  \cite{Davidoff_Brown_Diamond}.  Although life annuities provide income security in retirement, very few retirees choose a life annuity over a lump sum.  According to a recent survey exploring attitudes towards annuitization among individuals approaching retirement in the United Kingdom by \cite{Gardner_Wadsworth}, over half of the individuals in the sample chose not to annuitize given the option.  Whether the option was 100\% annuitization or only partial (50\%) annuitization, the attitude was the same. The dominant reason given for not wanting to annuitize in the survey is the preference for flexibility.  It is well known that annuity income is {\it not} reversible. In other words, annuity holders can neither surrender for a refund nor  short-sell (borrow against) their earlier purchased annuities, even when such a deal is desirable. 

In this paper,  we explore a way to add flexibility to life annuities by proposing a financial innovation, specifically a {\it reversible annuity}, an immediate life annuity with a surrender option.  The option to surrender allows an annuity holder to either borrow against or surrender any portion of her annuities at any time when she is still alive.  The purchase value of this reversible annuity is determined by the expected present value of future payments to the annuity holder, which follows the same principle as regular annuities.  The surrender value is set as a fixed proportion of  its purchase value at the time of surrendering. The surrender value can also be viewed as the purchasing value less a proportional surrender charge, which is a combination of transaction cost, operating charge, and compensation for adverse selection.  To explore how this reversible annuity would work for retirees as a reliable flow of  income, as well as an asset able to be surrendered under certain personal circumstance, we investigate the optimal investment strategy and optimal annuity purchase and surrender strategies for an individual who seeks to minimize the probability that she outlives her wealth, also called the probability of {\it lifetime ruin}. In other words, we assume that the retiree consumes at a exogenous level, and we determine the optimal investment strategy, as well as the optimal time to annuitize or to surrender, in order to minimize the probability that wealth will reach zero before her death. 

As a risk metric, the probability of lifetime ruin is widely used to investigate optimization problems  faced by retirees in a financial market.  This metric was first introduced by  \cite{Milevsky_Robinson} in a static environment  and was extended by \cite{Young} to a stochastic environment without immediate life annuities.  A recent paper by  \cite{Milevsky_Moore_Young} determined the optimal dynamic investment policy for an individual who consumes at a specific rate, who invests in a complete financial market, and who can buy irreversible immediate life annuities.  Milevsky, Moore, and Young show that the individual will not annuitize any of her wealth until she can fully cover her desired consumption with an immediate life annuity.  Additionally, \cite{Bayraktar_Young} investigate the optimal strategy for an retiree in a financial market with deferred (not immediate) life annuities.  Within the topic of minimizing probability of lifetime ruin in a complete financial market without life annuities, \cite{Bayraktar_Moore_Young} consider the case for which the exogenous consumption is random, and in \cite{Bayraktar_Young2}, consumption is ratcheted (that is, it is a non-decreasing function of maximum wealth).  \cite{Bayraktar_Young3} investigate the optimal strategy when consumption level is deterministic but borrowing is constrained. 

In contrast to the literature mentioned above, we allow an individual not only to buy an immediate life annuity, but also to surrender existing immediate life annuities with a proportional surrender charge.  This reversibility of life annuities and the incompleteness of the annuity market (due to the proportional surrender charge) creates a more complex optimization environment and makes the problem mathematically challenging.  Our model can be viewed as a generalization of the model by \cite{Milevsky_Moore_Young} in which annuities are irreversible,  and the limiting case for which the surrender value of existing annuity approaches zero is consistent with their study. 

Our work is the first to investigate the optimal strategies for a retiree in a market with reversible immediate life annuities. We comprehensively analyze the annuitization and investment strategies for such a retire.  We focus on how the proportional surrender charge, which ranges from $0\%$ to $100\%$ of the purchasing value of annuity,  affects an individual's optimal strategies.  We predict that, when the surrender charge is low enough,  the individual has incentive to annuitize partially. This distinguishes our model from the one with irreversible annuities, in which an individual is only willing to fully annuitize. This difference shows that the flexibility offered by reversible annuities might be able to resolve the ``annuity puzzle.'' 

The remainder of the paper is organized as follows:  In Section $2$, we present the financial market in which the individual invests her wealth.  In addition to investing in riskless and risky assets, the individual can purchase reversible immediate life annuities.  In Section $3$, we consider the life annuity as part of her total wealth, thereby allowing her assets to have negative value as long as the imputed surrender value of her annuity makes her total wealth positive.  We prove a verification theorem for the minimal probability of lifetime ruin in this case, and we obtain the minimal probability, along with optimal investment and annuitization strategies. In Section $4$, we consider the case for which individual is forced to keep the value of her riskless and risky assets non-negative (excluding the surrender value of the annuity) by surrendering the annuity when needed.  It turns out that the optimal annuitization strategy depends on the size of the proportional surrender charge.  We consider the case when the charge is large in Section $4.2.1$, and in Section $4.2.2$, we discuss the case when the charge is small.

\section{Minimizing the Probability of Lifetime Ruin}

\bigskip

In this section, we describe the financial market in which the individual can invest her wealth, and we formulate the problem of minimizing the probability of lifetime ruin in this market.  We allow the individual to purchase and surrender her reversible life annuity at any time.

\subsection{Financial model}

We consider an individual with future lifetime described by the random variable $\tau _{d}$.  Suppose $\tau _{d}$ is an exponential random variable with parameter $\lambda ^{S}$, also referred to as the force of mortality or
hazard rate; in particular, $\mathbb{E}[\tau _{d}]=1/\lambda ^{S}$.  The superscript $S$ indicates that the parameter equals the individual's subjective belief as to the value of her hazard rate.

We assume that the individual consumes wealth at a constant rate of $c\geq 0$; this rate might be given in real or nominal units.  One can interpret $c$ as the minimum net consumption level below which the individual cannot (or will not) reduce her consumption further; therefore, the minimum probability of lifetime ruin gives a lower bound for the probability of ruin under any consumption function bounded below by $c$.

The individual can invest in a riskless asset, which earns interest at the constant rate $r\geq 0$. Also, she can invest in a risky asset whose price satisfies
\begin{equation}
dS_{t}=\mu \, S_{t} \, dt + \sigma \, S_{t} \, dB_{t},\quad S_{0}=S>0,
\end{equation}%
\noindent in which $\mu >r$, $\sigma >0$, and $B$ is a standard Brownian motion with respect to a filtration $\mathbb{F}=\{\mathcal{F}_{t}\}$ of a probability space $(\Omega ,\mathcal{F},\P )$. We assume that $B$ is
independent of $\tau _{d}$, the random time of death of the individual. If $c $ is given as a real rate of consumption (that is, inflation adjusted), then we also express $r$ and $\mu $ as real rates.

Moreover, an individual can buy any amount of reversible immediate life annuity or surrender any portion of her existing annuity income and receive some fraction of its value. The purchase price of an immediate life annuity that pays \$1 per year continuously until the insured's death is given by
\begin{equation}
\bar{a}=\int_{0}^{\infty }e^{-rs}e^{-\lo s}\,ds=\frac{1}{r+\lo},  \label{imm-ann}
\end{equation}
in which $\lo>0$ is the constant objective hazard rate that is used to price annuities. In other words, in return for each \$$\overline{a}$ the individual pays for an immediate life annuity, she receives \$$1$ per year of continuous  annuity income until she dies.

Due to the reversibility of the life annuity, she can surrender any amount of the annuity she owns.  The surrender value for \$$1$ of annuity income is $(1-p) \overline{a}$ with $0<p\leq 1$.  The factor $p$ is the proportional surrender charge. In other words, the individual can get \$$(1-p) \overline{a}$ dollars from the issuer by giving up \$$1$ of annuity income.  Notice that the surrender value is less than the purchase price, and the difference is the surrender charge (in dollars).

Let $W_t$ denote the amount of wealth the individual has invested in the risky and riskless assets at time $t$, with $\pi _{t}$ in the risky asset and $W_t - \pi_t$ in the riskless.  Let $A_{t}^{+}$ denote the cumulative amount of annuity income bought on or before time $t$, and let $A_{t}^{-}$ denote the cumulative amount of annuity income
surrendered on or before time $t$.  Then, $A_{t}=A_{t}^{+}-A_{t}^{-}$ represents the cumulative amount of immediate life annuity income at time $t$. The investment and annuitization strategy $\{\pi _{t},A_{t} \}_{t\geq
0} $ is said to be \textit{admissible} if the processes $\{\pi _{t}\}_{t\geq 0}$ and $\{A_{t}^{\pm }\}_{t\geq 0}$ are adapted to the filtration $\mathbb{F}$, if $\int_{0}^{t}\pi_{s}^{2}\,ds<\infty $, almost surely, for all $t\geq 0$, and if $A_t \ge 0$, almost surely, for all $t \ge 0$.  The wealth dynamics of the individual for a given admissible strategy are given by
\begin{equation}
dW_{t}=[rW_{t-}+(\mu -r)\pi _{t-}-c+A_{t-}]dt+\sigma \pi _{t-}dB_{t}- \bar{a}dA_{t}^{+}+\bar{a}(1-p)dA_{t}^{-},\quad W_{0}=w\geq 0.
\label{eq:wealth}
\end{equation}

By \textquotedblleft lifetime ruin,\textquotedblright\  we mean that the individual's wealth reaches the line $w=-(1-p)\bar{a}A$ before she dies.  We denote the time of ruin by $\tau _{0}\triangleq \{t \ge 0:W_{t}^{\pi ,A} + (1-p)
\bar{a}A_{t}^{\pi ,A} \le 0\}$.  In Section $3$, we allow wealth (namely, the value of the riskless and risky assets) to be negative with the individual effectively borrowing against her annuity income.  Then, in Section 4, we require that wealth remain non-negative.  Note that $\tau _{0}$ is independent of $\tau _{d}$.  The minimum probability of lifetime ruin $\psi $ for the individual at time $0$ is defined by

\begin{equation}
\psi (w,A) \triangleq \inf_{\{\pi _{t},A_{t}\}}\P \left[\tau _{0}<\tau _{d}\Big|W_{0}=w,A_{0}=A,\tau _{d}>0,\tau _{0}>0\right].  \label{exp:ruin}
\end{equation}

\begin{remark}
Notice that because we assume that the hazard rates $\lambda ^{S}$ and $\lambda ^{O}$, as well as the financial parameters $r$, $\mu $, and $\sigma $, are constant, $\psi $ only depends on the state variables $w$ and $A$ and not upon time.
\end{remark}

\begin{remark}
We can derive an equivalent form for the minimum probability of ruin due to
the independence of the $\tau _{d}$ from $\tau _{0}:$ 
\begin{eqnarray}
\psi (w,A) &=&\inf_{\{\pi _{t},A_{t}\}}\P \left[\tau _{0}<\tau _{d}\Big| W_{0}=w,A_{0}=A,\tau _{d}>0,\tau _{0}>0\right]  \notag \\
&=&\inf_{\{\pi _{t},A_{t}\}}\mathbb{E}\left[\int_{0}^{\infty }\lambda^{S}e^{- \lambda ^{S}t} \, \mathbf{1}_{\{0\leq \tau _{0}\leq t\}} \, dt \, \Big| \, W_{0}=w,A_{0}=A,\tau _{d}>0,\tau _{0}>0\right]  \notag \\
&=&\inf_{\{\pi _{t},A_{t}\}}\mathbb{E}\left[\int_{\tau _{0}}^{\infty }\lambda^{S}e^{- \lambda ^{S}t} \, dt \,\Big| \, W_{0}=w,A_{0}=A,\tau _{d}>0,\tau _{0}>0\right]  \notag \\
&=&\inf_{\{\pi _{t},A_{t}\}}\mathbb{E}\left[e^{-\lambda ^{S}\tau _{0}} \, \Big| \, W_{0}=w,A_{0}=A,\tau _{d}>0,\tau _{0}>0\right] .  \label{exp:ruin1}
\end{eqnarray}
We will use this expression in our proof of verification theorem in next section.
\end{remark}

\begin{remark}
\cite{Milevsky_Moore_Young} show that if one only allows irreversible life annuities, then the individual will not buy a life annuity until her wealth is large enough to cover all her consumption.  Specifically, if $w\geq (c - A) \overline{a}$, then it is optimal for the individual to spend $(c - A)\overline{a}$ to buy an immediate annuity that will pay at the continuous rate $c - A$ for the rest of her life.  This income, together with the prior income of $A$, will cover her desired consumption rate of $c$.  In this case, the individual will not ruin, under the convention that if her net consumption rate becomes $c$, then she is not considered ruined even if her wealth is 0. (The latter occurs if her wealth is identically $(c-A)\overline{a}$ immediately before buying the annuity.)
\end{remark}

\section{NO BORROWING RESTRICTION}

In this section, we consider the case in which the individual's wealth $w$ is allowed to be negative, as long as $w + (1- p) \bar a A$ is positive.  Effectively, the individual is allowed to borrow against her life annuity income.

\subsection{\label{sec:i-motivation}Motivation for the Hamilton-Jacobi-Bellman Variational Inequality}

Let us first consider the strategies one can choose to minimize the  probability of ruin. Before ruin occurs or the individual dies, she can execute one or more of the following strategies:  (i) purchase additional annuity income, (ii) surrender existing annuity income, or (iii) do neither.

Now, suppose that at point $(w,A)$, it is optimal \textit{not} to purchase or surrender any annuity income.  In this case, we expect $\psi $ will satisfy the equation 
\begin{equation}\label{HJBVi}
\lambda ^{S}\psi =(rw-c+A)\psi _{w}+\min_{\pi }\left[ (\mu -r)\pi \psi _{w}+\frac{1}{2}\sigma ^{2}\pi ^{2}\psi _{ww}\right] . 
\end{equation}
Because the above policy is in general suboptimal, (\ref{HJBVi}) holds as an
inequality; that is, for all $(w,A)$, 
\begin{equation}
\lambda ^{S}\psi \leq (rw-c+A)\psi _{w}+\min_{\pi }\left[ (\mu -r)\pi \psi
_{w}+\frac{1}{2}\sigma ^{2}\pi ^{2}\psi _{ww}\right] .
\end{equation}

As we shall prove later, no continuous purchase of lifetime annuity income is optimal; that is, the problem of purchasing or surrendering annuity is one of singular control. Thus, if at the point $(w,A)$, it is optimal to purchase annuity income instantaneously, then the individual moves instantly from $(w,A)$ to $(w-\bar{a}\Delta A,A+\Delta A)$, for some $\Delta A>0$.  The optimality of this decision implies that 
\begin{equation}
\psi (w,A)=\psi (w-\bar{a}\Delta A,A+\Delta A),
\end{equation}
which in turn yields 
\begin{equation}
\bar{a}\psi _{w}(w,A) = \psi _{A}(w,A).  \label{buyi}
\end{equation}

Similarly, if it is optimal to surrender annuity income at the point $(w,A)$, the
following equation holds:
\begin{equation}
\psi (w,A)=\psi (w+(1-p)\bar{a}\Delta A,A-\Delta A),
\end{equation}
which implies
\begin{equation}  \label{selli}
(1-p)\bar{a}\psi_w(w,A) = \psi_A(w,A).
\end{equation}
Notice that the surrender value is a portion of the value of annuity determined by the proportional surrender charge $p$.

In general, such purchasing or surrendering policies are suboptimal; therefore, (\ref{buyi}) and (\ref{selli}) hold as inequalities and become
\begin{equation}  \label{ineqi1}
\bar{a}\psi_w(w,A) \le \psi_A(w,A),
\end{equation}
and
\begin{equation}
(1-p)\bar{a}\psi _{w}(w,A) \ge \psi _{A}(w,A).  \label{ineqi2}
\end{equation}
Because the individual will either buy additional annuity income, surrender existing annuity income, or neither, we expect that the probability of lifetime ruin is a solution of the following Hamilton-Jacobi-Bellman (HJB) variational inequality 
\begin{equation}
\begin{split}
\max \Bigg\{& \lambda ^{S}\psi - (rw-c+A)\psi _{w} - \min_{\pi }\left[ (\mu -r)\pi \psi _{w}+\frac{1}{2}\sigma ^{2}\pi ^{2}\psi _{ww}\right] , \\
& \bar{a}\psi _{w}(w,A)-\psi _{A}(w,A),\psi _{A}(w,A)-(1-p)\bar{a}\psi_{w}(w,A) \Bigg\}=0.
\end{split}
\label{HJBIV}
\end{equation}

Define $\dd{w_{s}(A)\triangleq \frac{pc}{\frac{1}{\bar{a}}-(1-p)r}-\bar{a}A}$, in which $A$ is the existing annuity income.  At the point $(w_{s}(A),A)$, suppose an individual borrows $\tilde w(A) \triangleq \dd{\frac{(c-A)\bar a - w_s(A)}{1 - r \bar a}}$ at the interest rate $r$.  She, then, has wealth $w_s(A) + \tilde w(A)$, which she spends to buy $\dfrac{1}{\bar{a}}(w_{s}(A)+\tilde w(A))$ additional life annuity income.  Therefore, the total annuity income she has is $A+\dfrac{1}{\bar{a}}
(w_{s}(A)+\tilde w(A)) = r \tilde w(A) + c$, which is just enough to cover the interest for the debt and the consumption and thereby ensure that lifetime ruin is impossible.  Note that $w_{s}(A)$ is the minimum required wealth to execute this strategy, so we call it the {\it safe level} for the case in which we allow wealth $w$ to be negative.  If asset and annuity income initially satisfy $w\geq w_{s}(A)$, then the individual will immediately execute this strategy to guarantee that her probability of lifetime ruin is zero.  It follows that
\begin{equation}
\psi (w,A)=0,
\end{equation}
for $w \ge w_s(A)$.

Recall that ruin occurs when $w+ (1-p)\bar{a} A \leq 0$, from which it follows that
\begin{equation}
\psi (w,A)=1,
\end{equation}
for  $w \le \underline{w}(A)\triangleq -(1-p)\bar{a}A$.

The two boundaries $w_{s}(A)$ and $\underline{w}(A)$ meet at $A=\dd{\frac{c}{1-(1-p) r \bar a}} > 0$ as in Figure \ref{fig:region}. Thus, it remains to solve the minimum probability of ruin in the region \hfill \break $\mathcal{D} \triangleq \left\{(w,A):\underline{w}(A)\leq
w\leq w_{s}(A), \, 0\leq A<\dfrac{c}{1-(1-p)r \bar a}\right\}$.

\subsection{Verification Theorem}

The discussion in Section $3.1$ motivates the following verification theorem:

\begin{thm}
\label{lem:verf-lemma} For any $\pi \in \mathbb{R},$ define the functional operator $\mathcal{L}^{\pi }$ through its action on a test function $f$ by
\begin{equation}
\L ^{\pi }f=[rw+(\mu -r)\pi -c+A]f_{w}+\frac{1}{2}\sigma ^{2}\pi
^{2}f_{ww}-\lambda ^{S}f,
\end{equation}
Let $v=v(w,A)$ be a non-increasing, non-negative, convex function of $w$ that is twice-differentiable with respect to $w,$ except possibly at $w=w_{s}(A)$ where we assume that it has right- and left-derivatives, and that is differentiable with respect to $A$. Suppose $v$ satisfies the following conditions on $\mathcal{D}$:
\begin{enumerate}
\item $\L ^{\pi }v(w,A)\geq 0$ for any $\pi \in \mathbb{R}$.

\item $\overline{a}v_{w}(w,A)-v_{A}(w,A)\leq 0$.

\item $(1-p)\overline{a}v_{w}(w,A)-v_{A}(w,A)\geq 0$.

\item $v(\underline{w}(A),A) = 1$, where $\underline{w}(A)$ is the lower
boundary of wealth for the problem.
\end{enumerate}
Then, 
\begin{equation}
v(w,A)\leq \psi (w,A),
\end{equation}
\noindent on $\mathcal D$.

\end{thm}

\begin{proof}
Suppose $\{ \pi_t \}$ is an admissible investment strategy, and define $\tau _{n}\triangleq \{t \geq 0:\int_{0}^{t}\pi _{t}^{2} \, dt\geq n\}$ and $\tau \triangleq \tau _{0}\wedge \tau _{n}$, which is a stopping time with respect to the filtration $\mathbb{F}$.   Then, by using It\^{o}'s formula for semi-martingales, we can write
\begin{equation}
\begin{split}
e^{- \lambda^{S} \tau}v(W_{\tau },A_{\tau })&
=v(w,A)+\int_{0}^{\tau }e^{- \lambda^{S} t}v_{w}(W_{t},A_{t}) \, \sigma \, \pi _{t}dB_{t} +\int_{0}^{\tau }e^{- \lambda ^{S} t} \L^{\pi_{t}}v(W_{t},A_{t})dt \\
& \quad +\int_{0}^{\tau }e^{- \lambda ^{S} t}\left[ v_{A}(W_{t},A_{t})- \bar{a}v_{w}(W_{t},A_{t})\right] d(A_{t}^{+})^{(c)} \\
& \quad +\int_{0}^{\tau }e^{- \lambda ^{S} t}\left[ (1-p)\bar{a}v_{w}(W_{t},A_{t})-v_{A}(W_{t},A_{t})\right] d(A_{t}^{-})^{(c)} \\
& \quad +\sum_{0\leq t\leq \tau }e^{- \lambda ^{S} t}\left[v(W_{t},A_{t})-v(W_{t-},A_{t-})\right] .
\end{split}
\label{eq:Ito}
\end{equation}

\noindent Here, $(A^{\pm })^{(c)}$ is the continuous part of $A^{\pm }$; that is, 
\begin{equation}
(A_{t}^{\pm })^{(c)}\triangleq A_{t}^{\pm }-\sum_{0\leq s\leq t}(A_{s}^{\pm }-A_{s-}^{\pm }).
\end{equation}

Since $v$ is non-increasing and convex in $w$, $v_{w}^{2}(w,A)\leq v_{w}^{2}(\underline{w}(A),A)$ for $w\geq \underline{w}(A)$. Therefore, 
\begin{equation}
\mathbb{E}\left[ \int_{0}^{\tau }e^{-2\lambda^{S}t} \, v_{w}^{2}(W_{t},A_{t}) \, \sigma ^{2} \, \pi _{t}^{2} \,dt \, \bigg| \, W_{0}=w,A_{0}=A\right] <\infty ,
\end{equation}
\noindent which implies that 
\begin{equation}
\mathbb{E}\left[ \int_{0}^{\tau }e^{-\lambda^{S} t} \, v_{w}(W_{t},A_{t}) \, \sigma \, \pi _{t} \, dB_{t} \, \bigg| \, W_{0}=w,A_{0}=A\right] =0.  \label{eq:sint}
\end{equation}

By taking expectations of equation (\ref{eq:Ito}), as well as using (\ref{eq:sint}) and Conditions 1, 2, and 3 in the statement of the theorem, we obtain 
\begin{equation}
\mathbb{E}\left[ e^{-\lambda^{S}\tau }v(W_{\tau },A_{\tau })\Big| W_{t}=w,A_{t}=A\right]\geq v(w,A).  \label{inequ}
\end{equation}
\noindent In deriving (\ref{inequ}), we also use the fact that 
\begin{equation}
\sum_{0\leq t\leq \tau }e^{-\lambda^{S}t}\left[
v(W_{t},A_{t})-v(W_{t-},A_{t-})\right] \geq 0,
\end{equation}

\noindent because Assumptions 2 and 3 imply that $v$ is non-decreasing in the direction of purchase and surrender.

Since $\tau_n \nearrow \infty$ and $v$ is bounded, applying the dominated convergence theorem to (\ref{inequ}) yields

\begin{equation}
\mathbb{E}\left[ e^{-\lambda ^{S}\tau _{0}}v(W_{\tau _{0}},A_{\tau _{0}}) \Big|W_{0}=w,A_{0}=A, \tau_d > 0, \tau_0 > 0 \right] \geq v(w,A).  \label{eq:v-T-v-t}
\end{equation}

\noindent By using Assumption $4$, one can rewrite (\ref{eq:v-T-v-t}) as

\begin{equation}
v(w,A)\leq \mathbb{E}\left[ e^{-\lambda ^{S}\tau _{0}}\Big|W_{0}=w,A_{0}=A, \tau_d > 0, \tau_0 > 0 \right] .  \label{eq:vleqanyprob}
\end{equation}
\noindent From this expression and from (\ref{exp:ruin1}), we infer that 
\begin{equation}
\begin{split}
v(w,A)& \leq \inf_{\{ \pi_t, A_t \} }\mathbb{E}\left[ e^{-\lambda ^{S}\tau_0 }\bigg|W_{0}=w,A_{0}=A, \tau_d > 0, \tau_0 > 0 \right] \\
& =\psi (w,A). 
\end{split}
\label{eq:vleq}
\end{equation}
\noindent

\end{proof}

We will use the following corollary of Theorm \ref{lem:verf-lemma} to determine $\psi $, the minimum probability of ruin, along with an optimal investment and annuitization strategy.

\begin{cor}
\label{cor:verf-lemma}
Suppose $v$ satisfies the conditions in Theorem \ref{lem:verf-lemma} and additionally is the probability of ruin associated with an admissible strategy, then $v=\psi $ on $D$ and the associated strategy is optimal.
\end{cor}

\subsection{\label{sec:linearization}Linearizing The Equation for $\protect \psi $ via Duality Arguments}

We hypothesize that in the region $\mathcal{D}\backslash \{w=w_{s}(A)$ or $w=\underline{w}(A)\}$ as defined in Section \ref{sec:i-motivation}, the optimal strategy for minimizing the probability of ruin is neither to purchase nor to surrender any life annuity income.  In other words, the individual does not buy any additional annuity income until her wealth reaches the safe level $w_{s}(A)$, which is consistent with the results of  \cite{Milevsky_Moore_Young}.  Additionally, the individual never surrenders her annuity income.  Intuitively, this makes sense because we count the annuity income's wealth equivalence in the ruin level $\underline w(A)$ and thereby allow the individual to borrow against future annuity income without actually forcing her to surrender the annuity.

Under this hypothesis, the first inequality in the HJB variational inequality (\ref {HJBIV}) holds with equality in the region $\mathcal{D}\backslash \{w=w_{s}(A)$ or $w=\underline{w}(A)\}$, and the minimum probability of ruin $\psi $ is the solution
to the following boundary-value problem (BVP) 
\begin{equation}
\lambda ^{S}\psi =(rw-c+A)\psi _{w}+\min_{\pi }\left[ (\mu -r)\pi \psi _{w} + \frac{1}{2}\sigma ^{2}\pi ^{2}\psi _{ww}\right],
\label{HJBi}
\end{equation}
with the boundary conditions
\begin{equation}
\psi (\underline{w}(A),A)=1,  \label{BCi1}
\end{equation}
and
\begin{equation}
\psi (w_s(A),A)=0.  \label{BCi2}
\end{equation}

\noindent After solving this BVP, we will show that its solution satisfies the conditions of the Verification Theorem
\ref{lem:verf-lemma} to verify our hypothesis.

To solve the BVP, we transform the nonlinear boundary value problem above into a linear free-boundary problem (FBP) via the Legendre transform.  Assume  $\psi (w,A)$ is convex with respect to $w$, which we verify later; therefore, we can define the concave
dual $\hat \psi$ of $\psi $ by
\begin{equation}
\hat{\psi}(y,A)=\min_{w\geq \underline{w}(A)}[\psi (w,A)+wy].  \label{Leg}
\end{equation}

\noindent The critical value $w^*(A)$ solves the equation $\psi_w(w, A) + y = 0$;
thus, $w^* = I(-y, A)$, in which $I$ is the inverse of $\psi_w$ with respect
to $w$. It follows that

\begin{equation}
\hat{\psi}_y (y, A)= I(-y, A),
\label{eq:1}
\end{equation}

\begin{equation}
\hat{\psi}_{yy}(y,A)= - \frac{1}{\psi_{ww}(w,A)}\Bigg|_{w=\psi_w^{-1}(-y,A)},
\end{equation}

\begin{equation}
\hat{\psi}_A(y,A)=\psi_A(w,A)\left|_{w=\psi_w^{-1}(-y,A)},\right.
\end{equation}
and
\begin{equation}
\hat{\psi}_{Ay}(y,A)=\psi_{Aw}(w,A)   \hat{\psi}_{yy}(y,A)\left|_{w=\psi_w^{-1}(-y,A)}.\right.
\label{eq:4}
\end{equation}

Rewrite the differential equation (\ref{HJBi}) in terms of $\hat{\psi}$ to get

\begin{equation}
-\lambda ^{S}\hat{\psi}-(r-\lambda ^{S})y\hat{\psi}_{y}+my^{2}\hat{\psi}_{yy}+y(c-A)=0,  \label{eq:FBP}
\end{equation}
in which $m=\dfrac{1}{2}\left( \dfrac{\mu -r}{\sigma }\right) ^{2}$ . The general solution of (\ref{eq:FBP}) is 
\begin{equation}
\hat{\psi}(y,A)=D_{1}(A)y^{B_{1}}+D_{2}(A)y^{B_{2}}+\frac{c-A}{r}y,
\label{eq:dFBP}
\end{equation}
in which 
\begin{equation}
B_{1,2}=\frac{1}{2m}\left( (r-\lambda ^{S}+m) \pm \sqrt{(r-\lambda^{S}+m)^{2}+4m\lambda ^{S}}\right),
\label{B}
\end{equation}
with $B_{1}>1$ and $B_{2}<0$.  It remains for us to determine the coefficients $D_{1}(A)$ and $D_{2}(A)$ via the two boundary conditions.

To that end, consider the boundary conditions (\ref{BCi1}) and (\ref{BCi2}).  Define 
\begin{equation}
\underline{y}(A)=-\psi _{w}(\underline{w}(A),A),  \label{y0i}
\end{equation}
and 
\begin{equation}
y_{s}(A)=-\psi _{w}(w_{s}(A),A).  \label{ybi}
\end{equation}
We will show later that $y_{s}(A) \leq \underline{y}(A)$, which is obvious if $\psi$ is decreasing and convex with respect to $w$. Then,  for the free boundaries $\underline y(A)$ and $y_s(A)$,  we obtain from (\ref{BCi1})  and (\ref{y0i}) 
\begin{equation}
\begin{cases}
\hat{\psi}(\underline{y}(A),A)=\psi (\underline{w}(A),A)+\underline{w}(A)\underline{y}(A)=1-(1-p)\bar{a}A\underline{y}(A), \\ 
\hat{\psi}_{y}(\underline{y}(A),A)=\underline{w}(A)=-(1-p)\bar{a}A;
\end{cases}
\label{eq:sellbdry}
\end{equation}
and from (\ref{BCi2}) and (\ref{ybi})
\begin{equation}
\begin{cases}
\hat{\psi}(y_{s}(A),A)=\psi (w_{s}(A),A)+w_{s}(A)y_{s}(A)= \left(\dfrac{pc}{\frac{1}{\bar{a}}-(1-p)r}-\bar{a}A\right)y_s(A), \\ 
\hat{\psi}_{y}(y_{s}(A),A)=w_s(A)=\dfrac{pc}{\frac{1}{\bar{a}}-(1-p)r}-\bar{a}A.
\end{cases}
\label{eq:buybdry}
\end{equation}

Next, we find $D_{1}(A)$ and $D_{2}(A)$ along with $\underline{y}(A)$ and $y_{s}(A)$. To do so, we use the four equations in (\ref{eq:sellbdry}) and (\ref{eq:buybdry}) to find these four unknowns in terms of $A$.  Substitute (\ref{eq:dFBP}) into (\ref{eq:sellbdry}) and (\ref{eq:buybdry}) to get 
\begin{equation}
D_{1}(A)\underline{y}(A)^{B_1}+D_{2}(A)\underline{y}(A)^{B_2}+{\frac{c-A}{r}}\underline{y}(A) =1-\frac{1-p}{r+\lambda ^{O}}A\underline{y}(A),
\label{eq:i1}
\end{equation}
\begin{equation}
D_{1}(A)B_{1}\underline{y}(A)^{B_{1}-1}+D_{2}(A)B_{2}\underline{y}(A)^{B_{2}-1}+{\frac{c-A}{r}} =-\frac{1-p}{r+\lambda ^{O}}A,  \label{eq:i2}
\end{equation}
\begin{equation}
D_{1}(A)y_{s}(A)^{B_1}+D_{2}(A)y_{s}(A)^{B_2}+{\frac{c-A}{r}}y_{s}(A) =\left( \frac{pc}{pr+\lambda ^{O}}-\frac{A}{r+\lambda ^{O}}\right) y_{s}(A),  \label{eq:i3}
\end{equation}
and
\begin{equation}
D_{1}(A)B_{1}y_{s}(A)^{B_{1}-1}+D_{2}(A)B_{2}y_{s}(A)^{B_{2}-1}+{\frac{c-A}{r}} =\frac{pc}{pr+\lambda ^{O}}-\frac{A}{r+\lambda ^{O}}.  \label{eq:i4}
\end{equation}

\noindent From (\ref{eq:i3}) and (\ref{eq:i4}), solve for $D_{1}(A)$ and $D_{2}(A)$ to obtain 
\begin{eqnarray}
D_{1}(A) &=&\frac{1-B_{2}}{B_{1}-B_{2}}\frac{1}{{y_{s}(A)}^{B_{1}-1}}\left( \frac{pc}{pr+\lambda ^{O}}-\frac{A}{r+\lambda ^{O}}-{\frac{c-A}{r}}\right) < 0,
\label{iD_1} \\
D_{2}(A) &=&\frac{B_{1}-1}{B_{1}-B_{2}}\frac{1}{{y_{s}(A)}^{B_{2}-1}}\left( \frac{pc}{pr+\lambda ^{O}}-\frac{A}{r+\lambda ^{O}}-{\frac{c-A}{r}}\right) < 0.
\label{iD_2}
\end{eqnarray}
Substituting $D_{1}(A)$ and $D_{2}(A)$ into (\ref{eq:i2}) gives 
\begin{equation}
\begin{split}
-\frac{1-p}{r+\lambda ^{O}}A& =\frac{B_{1}(1-B_{2})}{B_{1}-B_{2}}x(A)^{B_{1}-1}\left( \frac{pc}{pr+\lambda ^{O}}-\frac{A}{r+\lambda ^{O}}-{\frac{c-A}{r}}\right) \\
& \quad +\frac{B_{2}(B_{1}-1)}{B_{1}-B_{2}}x(A)^{B_{2}-1}\left( \frac{pc}{pr+\lambda ^{O}}-\frac{A}{r+\lambda ^{O}}-{\frac{c-A}{r}}\right) +\frac{c-A}{r}
\end{split}
\label{eq:**}
\end{equation}
with $x(A)\triangleq \underline{y}(A)/y_{s}(A)$ as a function of $A$.

With $A $ fixed, (i) if $x(A) = 1$, the right-hand side of (\ref{eq:**}) equals $\dfrac{pc}{pr+\lambda ^{O}}-\dfrac{A}{r+\lambda ^{O}}=w_{s}(A)>-\dfrac{1-p}{r+\lambda ^{O}}A=\underline{w}(A)$; (ii) if $x(A)\rightarrow +\infty $, then the right-hand side of (\ref{eq:**}) approaches $ -\infty $; and (iii) one can show that the right-hand side is strictly decreasing with respect to $x(A)$.  Therefore, there exists a unique $x(A) > 1$ that satisfies equation (\ref{eq:**}).

Substitute for $D_{1}(A)$ and $D_{2}(A)$ into (\ref{eq:i1}) to get 
\begin{equation}
\frac{1}{\underline{y}(A)}=\frac{c-A}{r}+\frac{1-p}{r+\lambda ^{O}}A+\left( \frac{pc}{pr+\lambda ^{O}}-\frac{A}{r+\lambda ^{O}}-\frac{c-A}{r}\right) \left( \frac{1-B_{2}}{B_{1}-B_{2}}x(A)^{B_{1}-1}+\frac{B_{1}-1}{B_{1}-B_{2}} x(A)^{B_{2}-1}\right) .  \label{eq:***}
\end{equation}
Then, by the definition of $x(A)$, the solution for $y_{s}(A)$ is simply 
\begin{equation}
y_{s}(A)=\frac{\underline{y}(A)}{x(A)} .
\label{eq:ys}
\end{equation}
Thus, we have solved the FBP given in (\ref{eq:FBP}), \eqref{eq:sellbdry}, and \eqref{eq:buybdry}, and we state this formally in the following proposition.

\begin{proposition}
\label{prop:psihat}The solution of the FBP (\ref{eq:FBP}) with conditions (\ref{eq:sellbdry}) and (\ref{eq:buybdry}) is given by (\ref{eq:dFBP}), with $D_{1}(A)$, $D_{2}(A)$, $\underline{y}(A)$, $y_{s}(A)$, and $x(A)$ defined in (\ref{iD_1}), (\ref{iD_2}), (\ref{eq:***}), (\ref{eq:ys}), and (\ref{eq:**}), respectively.
\end{proposition}

Next, we determine some properties of $\hat{\psi}(y,A)$; in particular, we show that it is concave.  Also, notice that we can rewrite the inequalities (\ref{ineqi1}) and (\ref {ineqi2}) in terms of $\hat{\psi}$ as
\begin{eqnarray}
\hat{\psi}_{A}(y,A) &\geq &-\frac{1}{r+\lambda ^{O}}y,  \label{dineqi1} \\
\hat{\psi}_{A}(y,A) &\leq &-\frac{1-p}{r+\lambda ^{O}}y,  \label{dineqi2}
\end{eqnarray}
for $y_{s}(A)\leq y \leq \underline{y}(A)$, and we show below that these inequalities hold for our solution $\hat \psi$.  

For notational simplicity, we drop the argument $A$ in $\underline{w}(A)$, $w_s(A)$, $\underline{y}(A)$, and $y_s(A)$ in much of the remainder of this subsection.  
By taking the derivative of (\ref{eq:**}) with respect to $A$, we get 
\begin{equation}
\begin{split}
& \left( w_{s}-\frac{c-A}{r} \right) \frac{(B_{1}-1)(1-B_{2})}{B_{1}-B_{2}}\left\{B_{1} x(A)^{B_{1}-1}-B_{2} x(A)^{B_{2}-1}\right\} \frac{dx(A)/dA}{x(A)} \\
& \quad = \left( \frac{1}{r} -\frac{1-p}{r+\lambda ^{O}} \right) - \left( \frac{1}{r} - \frac{1}{r+\lambda ^{O}} \right) \frac{\underline{w}-\dfrac{c-A}{r}}{w_{s}-\dfrac{c-A}{r}}.
\end{split}
\label{eq:dx}
\end{equation}
It is easy to check that the right-hand side of the equation above is $0$, which implies that 
\begin{equation}
\frac{dx(A)}{dA}=0.  \label{eq:x1}
\end{equation}
In other words, $x(A) = x$ is a constant, independent of $A$, and the equation (\ref{eq:**}) holds for all $A$ with the same value $x > 1$.

By taking the derivative of (\ref{eq:***}) and \eqref{eq:ys} with respect to $A$, we get 
\begin{equation}
\frac{dy_s(A)}{dA}=-y_s(A)\frac{\lambda ^{O}}{r(r+\lambda ^{O})}\frac{1}{w_s(A)-\dfrac{c-A}{r}}.
\label{eq:dy}
\end{equation}
Also, after substituting for $D_1(A)$ and $D_2(A)$ in \eqref{eq:dFBP}, we differentiate $\hat{\psi}(y, A)$ with respect to $A$ to get 
\begin{equation}
\begin{split}
\hat{\psi}_{A}(y,A)& =-y\left\{ \left( \frac{1}{r+\lambda ^{O}}-\frac{1}{r}\right) \left[ \frac{1-B_{2}}{B_{1}-B_{2}}\left( \frac{y}{y_{s}}\right) ^{B_{1}-1}+ \frac{B_{1}-1}{B_{1}-B_{2}}\left( \frac{y}{y_{s}}\right) ^{B_{2}-1}\right] + \frac{1}{r}\right\} \\
& \quad -\frac{dy_{s}(A)}{dA}\frac{(B_{1}-1)(1-B_{2})}{B_{1}-B_{2}}\left( w_s(A) -\frac{c-A}{r}\right) \left[
\left( \frac{y}{y_{s}}\right) ^{B_{1}}-\left( \frac{y}{y_{s}}\right)^{B_{2}}\right] .
\end{split}
\label{psi_A}
\end{equation}

\begin{proposition}
\label{prop:psihat_concave}$\hat{\psi}(y,A)$ given by Proposition \ref{prop:psihat} is concave with respect to $y$ and satisfies inequalities (\ref{dineqi1}) and (\ref{dineqi2}).
\end{proposition}

\begin{proof}
First, it is straightforward to show the positivity of $\underline{y}(A)$ from (\ref{eq:***}). This confirms that $\underline{y}(A) = y_s(A) x > y_{s}(A) > 0$ because $x > 1$.  It follows that $\hat{\psi}(y,A)$ is concave with respect to $y$ since both $D_{1}(A)<0$ and $D_{2}(A) < 0$, and both $B_1(B_1 - 1) > 0$ and $B_2 (B_2 - 1) > 0$.

To prove the inequalities, we first substitute (\ref{eq:dy}) into (\ref{psi_A}) to get 
\begin{equation}
\hat{\psi}_{A}(y,A)=y\frac{\lambda ^{O}}{r(r+\lambda ^{O})} \left[ \frac{B_{1}(1-B_{2})}{B_{1}-B_{2}}\left( \frac{y}{y_{s}}\right) ^{B_{1}-1}+\frac{(B_{1}-1)B_{2}}{B_{1}-B_{2}}\left( \frac{y}{y_{s}}\right) ^{B_{2}-1} \right] -\frac{y}{r},  \label{psi_A2}
\end{equation}
Substitute the expression for $\hat \psi_A(w, A)$ from \eqref{psi_A2} into inequalities (\ref{dineqi1}) and (\ref{dineqi2}) to obtain the equivalent inequalities
\begin{equation}
1\geq \frac{r+\lambda ^{O}}{r}-\frac{\lambda ^{O}}{r} \left[ \frac{B_{1}(1-B_{2})}{B_{1}-B_{2}}\left( \frac{y}{y_{s}}\right) ^{B_{1}-1}+\frac{(B_{1}-1)B_{2}}{B_{1}-B_{2}}\left( \frac{y}{y_{s}}\right) ^{B_{2}-1} \right] \geq 1-p.  \label{ieq}
\end{equation}

Notice that the first inequality holds with equality if $y=y_{s}(A)$ and the second inequality holds with equality if $y=\underline{y}(A)$.  Define the auxiliary
function 
\begin{equation}
f(z)=\frac{B_{1}(1-B_{2})}{B_{1}-B_{2}}z^{B_{1}-1}+\frac{(B_{1}-1)B_{2}}{B_{1}-B_{2}}z^{B_{2}-1},
\end{equation}
which is increasing for $\dd{1=\frac{y_{s}(A)}{y_{s}(A)}\leq z\leq \frac{\underline{y}(A)}{y_{s}(A)}=x}$.  Indeed, in this interval, 
\begin{equation}
f^{\prime }(z)= \frac{(B_1 - 1)(1 - B_2)}{B_1 - B_2} \left[ B_{1} z^{B_{1}-2} - B_{2} z^{B_{2}-2} \right] > 0.
\end{equation}
It follows that, for $y_{s}(A)\leq y\leq \underline{y}(A)$, the inequality (\ref{ieq}), and equivalently (\ref{dineqi1}) and (\ref{dineqi2}) hold.
\end{proof}

In the next section, we rely on the work in this section to show that the convex dual of $\hat{\psi }(y,A)$ equals the minimum probability of ruin $\psi(w,A) $.

\subsection{\label{sec:dual-dual}Relation Between the FBP and the Minimum Probability of Ruin}

\label{subsectioni} In this section, we show that the Legendre transform of the solution to the FBP given in (\ref{eq:FBP}), \eqref{eq:sellbdry}, and \eqref{eq:buybdry} is in fact the minimum probability of ruin $\psi $.  Since $\hat{\psi}$ is concave from Propostion \ref{prop:psihat_concave}, we can define its convex dual via the Legendre transform for $w\geq \underline{w}(A)$ as 
\begin{equation}
\Psi (w,A)=\max_{y\geq 0}[\hat{\psi}(y,A)-wy].  \label{3.57}
\end{equation}
Given $A$, the critical value $y^{\ast }$ solves the equation $\hat{\psi}_{y}(y,A)-w=0$. Thus $y^{\ast }(A)=I(w,A)$, in which $I$ is the inverse of $\hat{\psi}_{y}$.  In this case, we also have expressions similar to those in \eqref{eq:1}-\eqref{eq:4}.

Given $\hat{\psi}$, we proceed to find the boundary-value problem that $\Psi $ solves. In the partial differential equation for $\hat{\psi}$ in (\ref{eq:FBP}),
let $y=I(w,A)=-\Psi _{w}(w,A)$ to obtain 
\begin{equation}
\lambda ^{S}\Psi (w,A)=(rw-c)\Psi _{w}(w,A)-m\frac{\Psi _{w}^{2}(w,A)}{\Psi
_{ww}(w,A)}.  \label{3.63}
\end{equation}
Notice that we can rewrite (\ref{3.63}) as

\begin{equation}
\min_{\pi }\mathcal{L}^{\pi }\Psi =0,
\end{equation}
with the minimizing strategy $\pi^*$ given in feedback form by
\begin{equation}
\pi^* (w,A)=-\frac{\mu -r}{\sigma ^{2}}\frac{\Psi _{w}(w,A)}{\Psi _{ww}(w,A)}.
\end{equation}
Therefore, $\Psi $ satisfies Condition $1$ in Verification Theorem \ref{lem:verf-lemma}.

Next, consider the boundary conditions for (\ref{eq:FBP}).  First, the boundary conditions at $y_{s}(A)$, namely $\hat{\psi}(y_{s}(A),A)=w_{s}(A)y_{s}(A)$ and $\hat{\psi}_{y}(y_{s}(A),A)=w_{s}(A)$, imply that the corresponding dual value of $w$ is $w_{s}(A)$ and that 
\begin{equation}
\Psi (w_{s}(A),A)=0.  \label{3.64}
\end{equation}
Similarly, the boundary conditions at $\underline{y}(A)$, namely $\hat{\psi}(\underline{y}(A),A) = 1 + \underline w(A) \underline y(A)$ and $\hat{\psi}_{y}(\underline{y}(A),A)=\underline{w} (A)$, imply that the corresponding dual value of $w$ is $\underline{w}(A)$ and that 
\begin{equation}
\Psi (\underline{w}(A), A) = 1.  \label{3.65}
\end{equation}

Finally, Propostion \ref{prop:psihat_concave} implies that 
\begin{equation}
\bar{a}\Psi _{w}(w,A)-\Psi _{A}(w,A)\leq 0,  \label{3.66}
\end{equation}
and 
\begin{equation}
(1-p)\bar{a}\Psi _{w}(w,A)-\Psi _{A}(w,A)\geq 0.  \label{3.67}
\end{equation}
Therefore, $\Psi (w,A)$ satisfies Conditions $2$ and $3$ in Theorem \ref{lem:verf-lemma}.

From $\Psi _{w}(w,A) = -y^*(A)$ and the fact that $y \geq y_ s(A)>0$,  $\Psi(w,A)$ is decreasing with respect to $w$, and consequently $0 \leq \Psi(w,A) \leq 1$ for $(w,A) \in \mathcal{D}$ due to (\ref{3.64}) and (\ref{3.65}). Thus, $\Psi$ is the minimum probability of ruin by Corollary \ref{cor:verf-lemma}, and we state this formally in the next theorem.

\begin{thm} \label{thm:i}
The minimum probability of ruin $\psi(w, A)$ for $(w,A)\in \mathcal{D}$, in which $\mathcal D$ is defined by $\mathcal{D} = \left\{(w,A): \underline{w}(A)\leq w\leq w_{s}(A), 0\leq A<\dfrac{c}{1-(1-p)\frac{r}{r+\lambda ^{O}}} \right\}$ equals $\Psi (w,A)$ in \eqref{3.57}.  The associated optimal annuitization and investment strategies are given by

\begin{enumerate}
\item never to surrender existing annuity income;

\item to purchase additional annuity income only when wealth reaches $w_{s}(A)$, the safe level;

\item for $w\in \mathcal{D}\backslash \{w=w_{s}(A)$ or $w=\underline{w}(A)\}$, to invest the following amount of wealth in the risky asset:
\begin{equation*}
\pi ^{\ast }(w,A)=-\frac{\mu -r}{\sigma ^{2}}\frac{\psi _{w}(w,A)}{\psi_{ww}(w,A)}.
\end{equation*}
\end{enumerate}
\end{thm}

\subsection{\label{sec:num1}Numerical Examples}

In this section, we present numerical examples to demonstrate the results of Section \ref{sec:dual-dual}. We calculate the probabilities of lifetime ruin $\psi (w,A)$, as well as the associated investment strategies $\pi ^{\ast}(w,A)$ for different values of the existing annuity income $A$ and the surrender charge $p$.  We use the following  values of the parameters for our calculation:

\begin{itemize}
\item $\lambda ^{S}=\lambda ^{O}=0.04$; the hazard rate is such that the expected future lifetime is $25$ years.

\item $r=0.02$; the riskless rate of return is $2$\% over inflation.

\item $\mu =0.06$; the drift of the risky asset is $6$\% over inflation.

\item $\sigma =0.20$; the volatility of the risky asset is $20$\%.

\item $c=1$; the individual consumes one unit of wealth per year.
\end{itemize}

We focus on how the surrender penalty affects the probability of ruin and the optimal investment strategy. 

Figures \ref{fig:1.1}-\ref{fig:1.4} show the ruin probability $\psi (w,A)$ and the associated optimal investment  $\ \pi ^{\ast }(w,A)$ in the risky asset  with the parameters described above, as well as with values for $A$ and $p$ as indicated in the figures.  Each curve gives values from $w=\underline{w}(A)$ to $w=w_s(A)$,  in which $\underline{w}(A)$ and $w_{s}(A)$ vary with respect to $A$ and $p$. This is the reason why each curve lies in a distinct domain. From the figures, we can see that the proved properties are verified in these examples:   the probability of ruin is decreasing and convex with respect to $w$. We also observe that investment in the risky asset increases as wealth increases for each case.

\section{BORROWING RESTRICTION}

In this section, we consider the case in which the individual is forced to keep her wealth non-negative by surrendering the life annuity when needed.  With this restriction, the situation is different from the one we studied in the previous section because in this section, the individual cannot borrow against future life annuity income.  It is reasonable to apply this restriction because if the individual were to die, then the annuity income ceases.  Therefore, if the individual were to borrow against future annuity income and die, there might be insufficient assets available to pay the debt.

Therefore, ruin occurs when both an individual's annuity income $A$ and wealth $w$ are $0$ since she has no more annuity income to surrender to raise her wealth.  It follows that $\tau _{0}$ in this case reduces to the hitting time of $(w,A)=(0,0)$ because on the line $w=-(1-p)\bar{a}A$, $(0,0)$ is the only point at which wealth $w$ is non-negative.  Notice that the probability of lifetime ruin is not $1$ when wealth reaches $0$ if an individual still has existing annuity income, which differs from the case of irreversible annuities.

\subsection{HJB Variational Inequality and Verification Theorem}

As the preceding case without a borrowing restriction, we have the same HJB variational inequality because the individual still has only three options to minimize the probability of ruin: purchasing additional annuity income, surrendering existing annuity income, and doing neither. Suboptimality of each strategy, in general, is represented by an inequality, while the optimality of one's executed strategy at all time requires that at least one of the three inequalities holds as an equality.

We need only consider when $A<c$; otherwise the individual already has enough annuity income to cover her consumption and lifetime ruin is impossible.  In this case, the safe level is given by $w_{s}(A)\triangleq (c-A) \bar{a}$. When the individual's wealth reaches the safe level, she is able to purchase $(c-A)$ of additional annuity income and, thereby, ensure that lifetime ruin is impossible. Therefore,  we have the condition
\begin{equation}
\psi (w_{s}(A),A)=0.  \label{BCii1}
\end{equation}
Notice that, for a given existing annuity income $A$, more wealth is needed to reach the safe level if borrowing against the annuity is not allowed; that is, $w_s(A)$ in this section is greater than $w_s(A)$ in the previous section.

When the individual's wealth reaches $0$, she is forced to surrender her life annuity to keep her wealth non-negative.  In this case, an annuitization strategy $\{A_{t}\}$ is admissible if the associated wealth process $W_{t}\geq 0$. almost surely, for all $t \ge 0$.  Inspired by the optimal annuitization strategy obtained in Theorem \ref{thm:i} for the case in which borrowing is not restricted, we hypothesize that the individual will only surrender enough annuity income to keep wealth non-negative.  This means that on the boundary $w=0$, she executes instantaneous control, so we expect the following Neumann condition:
\begin{equation}
(1-p)\bar{a}\psi _{w}(0,A)=\psi _{A}(0,A).  \label{BCii2}
\end{equation}
Moreover, if both her wealth and annuity income are $0$, ruin occurs; that is, 
\begin{equation}
\psi (0,0)=1.  \label{BCii3}
\end{equation}

Therefore, we need to solve for $\psi (w,A)$ in the region $ \mathcal{D} \triangleq \{(w,A):0\leq w\leq w_{s}(A),0\leq A<c\}$. Notice that the safe level $w_{s}(A) = (c - A) \bar a$ is different from the previous case.   With $\mathcal{D}$ thus redefined, we obtain the same verification theorem and corollary as Theorem \ref{lem:verf-lemma} and Corollary \ref{cor:verf-lemma}, respectively.  Please refer to the previous section for details.

\subsection{Solving for $\psi $ via Duality Arguments}

Through the course of our study, we determined that the optimal annuitization strategy for the individual to minimize her probability of lifetime ruin depends on the value of $p$.  We will show that when the penalty for surrendering is greater than $p^{\ast }$, a critical value to be determined later, the individual will not purchase any annuity until her wealth reaches the safe level $w_s(A)$, at which point she buys annuity income to cover the shortfall $c-A$.  On the other hand, if the penalty is low enough, namely $p<p^{\ast }$, the individual has incentive to annuitize partially; that is, the individual purchases additional annuity to cover part of the shortfall $c-A$ when her wealth is strictly below the safe level.  In this case, the individual will keep some wealth to invest in the risky financial market and spend the surplus to purchase annuity income.  We solve for the minimum probability of lifetime ruin $\psi$ for the first case $p\geq p^*$ in Section \ref{sec:p-big} and for the second case $p<p^*$ in Section \ref{sec:p-small}.  We also obtain the corresponding optimal annuitization and investment strategies.

\subsubsection{\label{sec:p-big} $p\geq p^*$}

When $p\geq p^{\ast }$, we hypothesize that in the domain $\mathcal{D} \backslash \{w=w_{s}(A)$ or $w=0\}$, the
optimal strategy for minimizing the probability of ruin is neither to purchase nor to surrender any annuity income.  Under this hypothesis, the first inequality in the HJB variational inequality (\ref{HJBIV}) holds with equality, and the minimum probability of ruin $\psi $ is the solution to the following BVP
\begin{equation}
\lambda ^{S}\psi =(rw-c+A)\psi _{w}+\min_{\pi }\left[ (\mu -r)\pi \psi _{w}+ \frac{1}{2}\sigma ^{2}\pi ^{2}\psi _{ww}\right],
\label{HJBii}
\end{equation}
with boundary conditions
\begin{equation}
\psi (w_{s}(A),A)=0,
\label{eq:4.5}
\end{equation}
\begin{equation}
(1-p)\bar{a}\psi _{w}(0,A)=\psi _{A}(0,A),
\label{eq:4.6}
\end{equation}
and
\begin{equation}
\psi (0,0)=1.
\label{eq:4.7}
\end{equation}
After solving the BVP, we will show that its solution satisfies the conditions of the Verification Theorem \ref{lem:verf-lemma} to verify our hypothesis.

As in Section \ref{sec:linearization}, we can define a related linear free-boundary problem via the Legendre transform.  Specifically, for $(w,A)\in \mathcal{D}$, define 
\begin{equation}
\hat{\psi}(y,A)=\min_{w\geq 0}[\psi (w,A)+wy].
\end{equation}
We can rewrite (\ref{HJBii}) as 
\begin{equation}
-\lambda ^{S}\hat{\psi}-(r-\lambda ^{S})y\hat{\psi}_{y}+my^{2}\hat{\psi}_{yy}+y(c-A)=0.  \label{FBPii}
\end{equation}
Its general solution is
\begin{equation}
\hat{\psi}(y,A)=D_{1}(A)y^{B_{1}}+D_{2}(A)y^{B_{2}}+\frac{c-A}{r}y,
\label{dFBPii}
\end{equation}
with $B_1 > 1$ and $B_2 < 0$ defined in (\ref{B}).

Define 
\begin{equation}
y_{0}(A)=-\psi _{w}(0,A),  \label{y0ii}
\end{equation}
and 
\begin{equation}
y_{s}(A)=-\psi _{w}(w_{s}(A),A).  \label{ybii}
\end{equation}
We get from \eqref{eq:4.6} and (\ref{y0ii}) that 
\begin{equation}
\begin{cases}
\hat{\psi}_{A}(y_{0}(A),A)=-(1-p)\bar{a}y_{0}(A), \\ 
\hat{\psi}_{y}(y_{0}(A),A)=0;
\end{cases}
\label{sellbdryii}
\end{equation}
from \eqref{eq:4.5} and (\ref{ybii}) that 
\begin{equation}
\begin{cases}
\hat{\psi}(y_{s}(A),A) = (c - A) \bar a y_{s}(A), \\ 
\hat{\psi}_{y}(y_{s}(A),A) = (c - A) \bar a;
\end{cases}
\label{buybdryii}
\end{equation}
and from \eqref{eq:4.7} and (\ref{y0ii}) that
\begin{equation}
\hat{\psi}(y_{0}(0),0)=1.
\label{ruinbdryii}
\end{equation}

Next, we determine $D_{1}(A)$ and $D_{2}(A)$ along with $y_0(A)$ and $y_{s}(A)$.  Rewrite (\ref{sellbdryii}), (\ref{buybdryii}), and \eqref{ruinbdryii} using (\ref{dFBPii}) to get 
\begin{eqnarray}
D_{1}(A)B_{1}y_{0}(A)^{B_{1}-1}+D_{2}(A)B_{2}y_{0}(A)^{B_{2}-1}+{\frac{c-A}{r}} &=&0,  \label{eq:ii1} \\
D_{1}^{\prime }(A)y_{0}(A)^{B_{1}-1}+D_{2}^{\prime }(A)y_{0}(A)^{B_{2}-1} &=& {\frac{1}{r}}-{\frac{1-p}{{r+\lambda ^{O}}}},  \label{eq:ii2} \\
D_{1}(A)B_{1}y_{s}(A)^{B_{1}-1}+D_{2}(A)B_{2}y_{s}(A)^{B_{2}-1} &=& {\frac{c-A}{r+\lambda ^{O}} - {\frac{c-A}{r}}},  \label{eq:ii3} \\
D_{1}(A)y_{s}(A)^{B_{1}-1}+D_{2}(A)y_{s}(A)^{B_{2}-1} &=&{\frac{c-A}{r+\lambda ^{O}}}-{\frac{c-A}{r}},  \label{eq:ii4} \\
D_{1}(0)y_{0}(0)^{B_{1}}+D_{2}(0)y_{0}(0)^{B_{2}}+{\frac{c}{r}}y_0(0) &=&1.
\label{eq:ii5}
\end{eqnarray}
From (\ref{eq:ii3}) and (\ref{eq:ii4}), we get 
\begin{eqnarray}
D_{1}(A) &=&-\frac{1-B_{2}}{B_{1}-B_{2}}\frac{\lambda ^{O}}{r(r+\lambda ^{O})}(c-A)\frac{1}{y_{s}(A)^{B_{1}-1}} < 0,  \label{iid1} \\
D_{2}(A) &=&-\frac{B_{1}-1}{B_{1}-B_{2}}\frac{\lambda ^{O}}{r(r+\lambda ^{O})}(c-A)\frac{1}{y_{s}(A)^{B_{2}-1}} < 0.  \label{iid2}
\end{eqnarray}
Then, substitute $D_{1}(A)$ and $D_{2}(A)$ into (\ref{eq:ii1}) to get 
\begin{equation}
\frac{\lambda ^{O}}{r+\lambda ^{O}}\left[ \frac{B_{1}(1-B_{2})}{B_{1}-B_{2}} \left( \frac{y_{0}(A)}{y_{s}(A)}\right) ^{B_{1}-1}+\frac{B_{2}(B_{1}-1)}{ B_{1}-B_{2}}\left( \frac{y_{0}(A)}{y_{s}(A)}\right) ^{B_{2}-1}\right] =1.
\label{iix}
\end{equation}
It is clear that $\dfrac{y_{0}(A)}{y_{s}(A)}$ is independent of $A$, and one can show that it is greater than $1$ through an  argument similar to the one following (\ref{eq:**}).  So, we define the constant 
\begin{equation}
x\triangleq \frac{y_{0}(A)\text{ }}{y_{s}(A)} . \label{iix1}
\end{equation}
Now, differentiate (\ref{iid1}) and (\ref{iid2}) with respect to $A$ and substitute into (\ref{eq:ii2}) to get 
\begin{equation}
\begin{split}
& \frac{dy_{s}(A)}{dA}(c-A)\frac{\lambda ^{O}}{r(r+\lambda ^{O})}\frac{(B_{1}-1)(1-B_{2})}{B_{1}-B_{2}}\left( x^{B_{1}}-x^{B_{2}}\right) \\
& \quad = -\frac{1-p}{r+\lambda ^{O}}xy_{s}(A)-y_{s}(A)\left\{ \frac{\lambda ^{O}}{r(r+\lambda ^{O})}\left[ \frac{1-B_{2}}{B_{1}-B_{2}}x^{B_{1}}+\frac{B_{1}-1}{B_{1}-B_{2}}x^{B_{2}}\right] -\frac{x}{r}\right\} .
\label{eq:4.25}
\end{split}
\end{equation}
Solve (\ref{iix}) for $x^{B_{2}-1}$ to simplify \eqref{eq:4.25} and obtain 
\begin{equation}
\frac{1}{y_{s}(A)}\frac{dy_{s}(A)}{dA} = \frac{K}{c-A},  \label{iidy}
\end{equation}
in which
\begin{equation}
K=\frac{-\dfrac{B_{2}}{1-B_{2}}\dfrac{1-p}{r+\lambda ^{O}}+\dfrac{\lambda ^{O}}{r(r+\lambda ^{O})}x^{B_{1}-1}-\dfrac{1}{r}}{\dfrac{\lambda ^{O}}{r(r+\lambda^{O})}x^{B_{1}-1}-\dfrac{1}{r}}
\label{iiK}
\end{equation}

Define the critical value $p^{\ast }$ as follows:
\begin{equation}
p^{\ast }\triangleq \frac{1}{B_{2}}-\frac{1-B_{2}}{B_{2}}\frac{\lambda ^{O}}{r} \left(x^{B_{1}-1}-1 \right).
\label{eq:pstar}
\end{equation}
It is straightforward to show that $K\geq 0$ iff $p\geq p^{\ast }$.  As we mentioned, we only consider the case $p\geq p^{\ast }$ here and leave the discussion for $p<p^*$ in Section \ref{sec:p-small}.  The expressions in \eqref{iix1} and
(\ref{iidy}) imply that
\begin{equation}
y_{0}(A) = \left(\frac{c}{c-A} \right)^{K}y_{0}(0),  \label{iiy0}
\end{equation}
and
\begin{equation}
y_{s}(A)=\frac{y_{0}(A)\text{ }}{x}.  \label{iiys}
\end{equation}

We determine the value of $y_{0}(0)$ by substituting (\ref{iid1}) and (\ref{iid2}) into (\ref{eq:ii5}):
\begin{equation}
\frac{1}{y_{0}(0)}=  \frac{c}{r} \left[ 1-\frac{\lambda ^{O}}{r+\lambda ^{O}}\frac{1-B_{2}}{B_{1}-B_{2}}x^{B_{1}-1}-\frac{\lambda ^{O}}{r+\lambda ^{O}}\frac{B_{1}-1}{B_{1}-B_{2}}x^{B_{2}-1}\right] .
\label{iiy00}
\end{equation}
By solving for $x^{B_2-1}$ from (\ref{iix}) and substituting it into (\ref{iiy00}), we get
\begin{equation}
\frac{1}{y_0(0)} = \frac{c}{r} \left( - \frac{1 - B_2}{B_2} \right) \left( 1-\frac{\lambda^O}{r+\lambda^O} \, x^{B_1-1} \right) >0.
\label{ineq:y00}
\end{equation}
The inequality in \eqref{ineq:y00} holds because $x^{B_1 - 1} < \left(r + \lo \right)/\lo$, which is straightforward to show from equation \eqref{iix} and the fact that the left-hand of that equation is increasing with respect to $x$.  From this inequality, we conclude that both $y_s(A)$ and $y_0(A)$ are positive for $(w,A)\in \mathcal{D}$.

\begin{proposition}
\label{prop:psihatii}The solution $\hat{\psi}(y,A)$ for the FBP \eqref{FBPii} with conditions \eqref{sellbdryii}, \eqref{buybdryii}, and \eqref{ruinbdryii} is given by \eqref{dFBPii}, with $D_{1}(A)$,  $D_{2}(A)$, $y_{0}(0)$, $y_{0}(A)$, $y_{s}(A)$, $x$, and $K$ defined in \eqref{iid1}, \eqref{iid2}, \eqref{iiy00}, \eqref{iiy0}, \eqref{iiys}, \eqref{iix}, and \eqref{iiK}, respectively.
\end{proposition}

Notice that we can rewrite the inequalities (\ref{ineqi1})
and (\ref{ineqi2}) in terms of $\hat{\psi}$ as 
\begin{eqnarray}
\hat{\psi}_{A}(y,A) &\geq &-\frac{1}{r+\lambda ^{O}}y,  \label{dineqii1} \\
\hat{\psi}_{A}(y,A) &\leq &-\frac{1-p}{r+\lambda ^{O}}y.  \label{dineqii2}
\end{eqnarray}

\begin{proposition}
\label{prop:psihat_concaveii}
$\hat{\psi}(y,A)$ given by Proposition \ref{prop:psihatii} is concave and satisfies inequalities (\ref{dineqii1}) and (\ref{dineqii2}).
\end{proposition}

\begin{proof}
The proof of the concavity of $\hat \psi$ with respect to $y$ follows from the observations that both $D_1(A) < 0$ and $D_2(A) < 0$ and that both $B_1(B_1 - 1) > 0$ and $B_2(B_2 - 1) > 0$.

To prove the inequalities, differentiate (\ref{iid1}) and (\ref{iid2}) with respect to $A$, substitute those expressions into $\hat{\psi}_{A}(y,A)$, use (\ref{iidy}) to simplify, and obtain 
\begin{equation}
\begin{split}
\hat{\psi}_{A}(y,A)& =yK\frac{\lambda ^{O}}{r(r+\lambda ^{O})}\frac{(B_{1}-1)(1-B_{2})}{B_{1}-B_{2}}\left[ \left( \frac{y}{y_{s}(A)}\right)^{B_{1}-1}-\left( \frac{y}{y_{s}(A)}\right) ^{B_{2}-1}\right] \\
& \quad +y\left\{ \frac{\lambda ^{O}}{r(r+\lambda ^{O})}\left[ \frac{1-B_{2}}{B_{1}-B_{2}}\left( \frac{y}{y_{s}(A)}\right) ^{B_{1}-1}+\frac{B_{1}-1}{B_{1}-B_{2}}\left( \frac{y}{y_{s}(A)}\right) ^{B_{2}-1}\right] -\frac{1}{r} \right\} .
\end{split}
\end{equation}
Then, rewrite inequalities (\ref{dineqii1}) and (\ref{dineqii2}) in the equivalent form as 
\begin{equation}
\begin{split}
1& \geq -K\frac{\lambda ^{O}}{r}\frac{(B_{1}-1)(1-B_{2})}{B_{1}-B_{2}}\left[
\left( \frac{y}{y_{s}(A)}\right) ^{B_{1}-1}-\left( \frac{y}{y_{s}(A)}\right)
^{B_{2}-1}\right] \\
& \quad -\frac{\lambda ^{O}}{r}\left[ \frac{1-B_{2}}{B_{1}-B_{2}}\left( \frac{y}{y_{s}(A)}\right) ^{B_{1}-1}+\frac{B_{1}-1}{B_{1}-B_{2}}\left( \frac{y}{y_{s}(A)}\right) ^{B_{2}-1}\right] +\frac{r+\lambda ^{O}}{r}\geq 1-p.
\end{split}
\label{iieq}
\end{equation}
To prove (\ref{iieq}), define the function $g$ by 
\begin{equation*}
g(z)=-K\frac{(B_{1}-1)(1-B_{2})}{B_{1}-B_{2}}\left[ z^{B_{1}-1}-z^{B_{2}-1}\right] -\left[ \frac{1-B_{2}}{B_{1}-B_{2}}z^{B_{1}-1}+\frac{B_{1}-1}{B_{1}-B_{2}}z^{B_{2}-1}\right] .
\end{equation*}
For $z\geq 1$, $g$ is decreasing because 
\begin{equation} \label{eq:gdecr}
\begin{split}
g^{\prime }(z)& =-K\frac{(B_{1}-1)(1-B_{2})}{B_{1}-B_{2}}\left[(B_{1}-1)z^{B_{1}-2}+(1-B_{2})z^{B_{2}-2}\right] \\
& \quad -\frac{(B_{1}-1)(1-B_{2})}{B_{1}-B_{2}}\left[ z^{B_{1}-2}-z^{B_{2}-2}
\right] \leq 0.
\end{split}
\end{equation}
Also, the first inequality in \eqref{iieq} holds with equality when $y = y_s(A)$, and the second inequality holds with equality when $y = y_0(A)$.  Therefore, (\ref{iieq}) holds for $y_{s}(A)\leq y\leq y_{0}(A)$.
\end{proof}

Since $\hat{\psi}$ is concave, we can define its convex dual via the Legendre transform:
\begin{equation}
\Psi (w,A)=\max_{y\geq y_s(A)} \left[\hat{\psi}(y,A)-wy \right].  \label{cpsi}
\end{equation}
As in Section \ref{sec:dual-dual}, we can prove that $\Psi$ is the minimum probability of ruin $\psi$, and we have the
following theorem.

\begin{thm} \label{thm:4.1}
When $p\geq p^{\ast }$ and the borrowing restriction is enforced, the minimum probability of ruin $\psi(w, A)$ for $(w,A)\in \mathcal{D}$, in which $\mathcal D$ is defined by $ \mathcal{D} = \{(w,A):0\leq w\leq w_{s}(A),0\leq A<c\}$, is given by $\Psi (w,A)$ in \eqref{cpsi}. The associated optimal annuitization and investment strategies are given by

\begin{enumerate}
\item to surrender existing annuity income instantaneously to keep wealth non-negative as needed;

\item to purchase additional annuity income only when wealth reaches $w_{s}(A)$, the safe level;

\item for $w\in \mathcal{D}\backslash \{w=w_{s}(A)\}$, to invest the following amount of wealth in the risky asset:
\begin{equation}
\pi ^*(w,A)=-\frac{\mu -r}{\sigma ^{2}}\frac{\psi _{w}(w,A)}{\psi _{ww}(w,A)}.  \label{iipi}
\end{equation}
\end{enumerate}
\end{thm}

It is clear from Theorem \ref{thm:4.1} that the optimal annuitization strategy is independent of the surrender charge $p$ as long as $p \ge p^*$.  However, it is not clear how the optimal investment strategy and the minimum probability of ruin vary with $p$.  We investigate this in the next proposition.

\begin{proposition} \label{prop:iipi}
$\pi ^{\ast }(w,A)$ given in \eqref{iipi} is independent of the surrender charge $p$, and the probability of ruin $\psi(w, A)$  increases with respect to $p$.
\end{proposition}

\begin{proof}
Fix $w$ and $A$. Given $w$, the corresponding $y$ is defined by (\ref{cpsi}) as
\begin{equation}
w=\hat{\psi}_{y}(y,A),  \label{temp1}
\end{equation}
which implies that $\psi _{w}(w,A)=-y$ and $\psi_{ww}(w,A) = -1/\hat{\psi}_{yy}(y,A)$.
Thus, we can write the optimal investment amount as
\begin{equation}
\pi ^{\ast }(w,A)= - \frac{\mu -r}{\sigma ^{2}}y\hat{\psi}_{yy}(y,A).
\label{temp2}
\end{equation}
By substituting (\ref{dFBPii}), (\ref{iid1}), (\ref{iid2}), (\ref{iiy0}), and (\ref{iiys}) into (\ref{temp1}) and (\ref{temp2}), we get the following two expressions, respectively:
\begin{equation}
w= \widetilde{D}_{1}(A)\left[ \left( \frac{c-A}{c}\right) ^{K}y\right]^{B_{1}-1} + \widetilde{D}_{2}(A)\left[ \left( \frac{c-A}{c}\right) ^{K}y\right] ^{B_{2}-1}+\frac{c-A}{r},
\label{iiw}
\end{equation}
and 
\begin{equation}
\pi ^*(w,A) = - \frac{\mu -r}{\sigma ^{2}} \left\{ (B_{1}-1)\widetilde{D}_{1}(A)\left[ \left( \frac{c-A}{c}\right) ^{K}y\right] ^{B_{1}-1}+  (B_{2}-1)\widetilde{D}_{2}(A)\left[ \left( \frac{c-A}{c}\right) ^{K}y\right] ^{B_{2}-1} \right\},  \label{iipi1}
\end{equation}
in which 
\begin{eqnarray}
\widetilde{D}_{1}(A) &=&-\frac{B_1(1-B_{2})}{B_{1}-B_{2}}\frac{\lambda ^{O}}{r(r+\lambda ^{O})}(c-A)\left( \frac{x}{y_{0}(0)}\right) ^{B_{1} - 1}, \\
\widetilde{D}_{2}(A) &=&-\frac{B_2(B_{1}-1)}{B_{1}-B_{2}}\frac{\lambda ^{O}}{r(r+\lambda ^{O})}(c-A)\left( \frac{x}{y_{0}(0)}\right) ^{B_{2} - 1}.
\end{eqnarray}
The numbers $x$ and $y_{0}(0)$ are independent of $p$ by (\ref{iix}) and (\ref{iiy00}), respectively.  Thus, $\widetilde{D}_{1}(A)$ and $\widetilde{D}_{2}(A)$ are also independent of $p$.  From (\ref{iiw}), we deduce that  $z = \left( \dfrac{c-A}{c}\right) ^{K}y$, which determines $\pi^*(w,A)$ via (\ref{iipi1}), does not depend on $p$.  Indeed, differentiate \eqref{iiw} with respect to $p$ to obtain
\begin{equation}
\begin{split}
0 &= \left[ \widetilde D_1(A) (B_1 - 1) z^{B_1 - 2} + \widetilde D_2(A) (B_2 - 1) z^{B_2 - 2} \right] \frac{\partial z}{\partial p} \\
&= \hat \psi_{yy}(y, A) \, \frac{x}{y_0(0)} \, \frac{\partial z}{\partial p}.
\end{split}
\end{equation}
Because $\hat \psi$ is strictly concave with respect to $y$ for $y_s(A) \le y \le y_0(A)$, it follows that $ \frac{\partial z}{\partial p} = 0$, from which we deduce that $z = \left( \dfrac{c-A}{c}\right) ^{K}y$ is independent of $p$.  Therefore, the optimal investment strategy $\pi^*(w,A)$ does not depend on $p$.

Next, we show that $\partial \psi(w, A)/ \partial p > 0$.  To this end, recall from \eqref{cpsi} that
\begin{equation}
\begin{split}
\psi(w, A) &= \hat \psi(y, A) - w y \\
&= \frac{x}{y_0(0)} \left[ \widetilde D_1(A) z^{B_1} + \widetilde D_2(A) z^{B_2} \right] + \left( \frac{c - A}{r} - w \right)y,
\end{split}
\end{equation}
in which $y$ is given by \eqref{temp1}.  Differentiate this expression with respect to $p$ to obtain
\begin{equation}
\begin{split}
\frac{\partial \psi(w, A)}{ \partial p} &=  \left( \frac{c - A}{r} - w \right) \frac{\partial y}{ \partial p} \propto \frac{\partial y}{ \partial p}  \\
& = - \ln \left( \frac{c - A}{c} \right) \, \frac{\partial K}{\partial p} \, y > 0,
\end{split}
\end{equation}
in which we use the fact that $z  = \left( \dfrac{c-A}{c}\right) ^{K}y$ is independent of $p$ in order to compute $\partial y/\partial p$, and we use the definition of $K$ in \eqref{iiK} to deduce that $\partial K/\partial p$ is positive.  Thus, the probability of ruin $\psi(w, A)$ increases as $p$ increases.
\end{proof}

\begin{remark}
Proposition \ref{prop:iipi} indicates that, when borrowing is restricted and $p\geq p^{\ast} $, an individual follows exactly the same investment and annuitization strategies regardless of the value of $p \ge p^*$.  The individual makes her decision based on her wealth and existing annuity income only.  It is not surprising that for given values of $w$ and $A$, the probability of ruin is smaller for a smaller $p$ because with a smaller surrender charge $p$, one receives more wealth when surrendering a given amount of annuity income.
\end{remark}

In this section, we determined the optimal annuitization and investment strategies and the corresponding minimum probability of ruin under the condition $p\geq p^{\ast }$.  The latter is equivalent to the condition $K\geq 0$, which plays a critical role in the proof of Proposition \ref{prop:psihat_concaveii}.  If $K$ were negative, then inequality (\ref{dineqii1}) would not hold for $y$ just above $y_{s}(A)$.  Consequently, $\Psi (w,A)$ would not satisfy Condition $2$ in the Verification Theorem \ref{lem:verf-lemma}.   From this, we infer that buying additional annuity income before reaching the safe level $w=w_{s}(A)$ might be optimal when $p<p^*$.  With this in mind, we proceed to the next section.

\subsubsection{\label{sec:p-small}$p<p^*$}

In this section, we consider the case for which $p<p^{\ast }$.  Define $\mathcal{D}_{1}\triangleq \{(w,A):0\leq w\leq w_{b}(A),0\leq A<c\}$ with $ w_{b}(A) \in [0, w_s(A)]$ to be specified later.  Also, define $\mathcal{D}_{2}\triangleq \{(w,A):w_{b}(A)< w \le w_{s}(A),0\leq A<c\}$, and note that $\mathcal{D} = \mathcal{D}_{1}\cup \mathcal{D}_{2}$.  As in the case for which $p \ge p^*$ in Section \ref{sec:p-big}, we only need to determine the minimum probability $\psi (w,A)$ for $(w,A)\in \mathcal{D}$.

We hypothesize that the following annuitization strategy is optimal:  If $(w,A)\in \mathcal{D}_{1}\backslash \{w=0$ or $w=w_{b}(A)\}$, the individual neither purchases or surrenders any life annuity income.  If $(w,A)\in \mathcal{D}_{2}$, the individual purchases just enough annuity income to reach the region $\mathcal{D}_{1}$.  That is to say, if she starts
with $(w,A)\in D_{2}$, the optimal strategy is to purchase $\Delta A$ of annuity income such that $w-{\Delta A}/(r+\lambda ^{O}) = w_{b}(A+\Delta A)$.  Thereafter, whenever wealth reaches the barrier $w_{b}(A)$, she keeps her portfolio of wealth and annuity income $(w,A)$ in the region $\mathcal{D}_{1}$ by instantaneously purchasing enough annuity income.  On the other hand, when wealth reaches $0$, the individual instantaneously surrenders enough annuity income to keep her wealth non-negative as we hypothesized in Section \ref{sec:p-big}. 

Ruin occurs only when $(w,A)=(0,0)$, at which point one has no existing annuity income to surrender to keep wealth non-negative.  Under the hypothesis for the optimal annuitization strategy, we anticipate that the associated minimum probability of ruin $\psi $ satisfies the following boundary-value problem.  After we solve this BVP, we will verify our hypothesis via Verification Theorem \ref{lem:verf-lemma}.

\begin{enumerate}
\item For $(w,A)\in \mathcal{D}_{1}$, $\psi(w,A)$ solves the following BVP:
\begin{equation}
\lambda ^{S}\psi =(rw-c+A)\psi _{w}+\min_{\pi }\left[ (\mu -r)\pi \psi _{w} + \frac{1}{2}\sigma ^{2}\pi ^{2}\psi _{ww}\right],\label{HJBiii}
\end{equation}
with boundary conditions
\begin{equation}
\bar{a}\psi _{w}(w_{b}(A),A)=\psi _{A}(w_{b}(A),A),  \label{BCiii1}
\end{equation}
\begin{equation}
(1-p)\bar{a}\psi _{w}(0,A)=\psi _{A}(0,A),  \label{BCiii2}
\end{equation}
and
\begin{equation}
\psi (0,0)=1.  \label{BCiii3}
\end{equation}
\item For $(w,A)\in \mathcal{D}_{2}~$, we have 
\begin{equation}
\psi (w,A)=\psi \left(w-\frac{\Delta A}{r+\lambda ^{O}},A+\Delta A \right),
\label{2ndregion}
\end{equation}
in which $w - {\Delta A}/(r+\lambda ^{O}) = w_{b}(A + \Delta A)$. Notice that $(w -{\Delta A}/(r+\lambda ^{O}), A+\Delta A)\in \mathcal{D}_{1}$, and thus $\psi (w - {\Delta A}/(r+\lambda ^{O}), A+\Delta A)$ is determined by the BVP (\ref{HJBiii})-(\ref{BCiii3}).
\item To solve for $\psi $ in the entire region $\mathcal{D}$, as well as to determine the purchase boundary $w_b(A)$, we also rely on a smooth fit condition across the boundary $w_{b}(A)$, namely,
\begin{equation}
\bar{a}\psi _{ww}(w_{b}(A),A)=\psi _{wA}(w_{b}(A),A).  \label{smfi}
\end{equation}
\end{enumerate}

\bigskip We first consider $\psi (w,A)$ in the region $\mathcal{D}_{1}$ by solving the related BVP \eqref{HJBiii}-\eqref{BCiii3}.  Hypothesize that $\psi$ is convex with respect to $w$, and define its concave dual via the Legendre transform by
\begin{equation}
\hat{\psi}(y,A) = \min_{w\geq 0} \left[\psi (w,A)+wy \right].
\end{equation}
As before, rewrite (\ref{HJBiii}) as 
\begin{equation}
-\lambda ^{S}\hat{\psi}-(r-\lambda ^{S})y\hat{\psi}_{y}+my^{2}\hat{\psi}_{yy} + y(c-A)=0.  \label{FBPiii}
\end{equation}
Its general solution is 
\begin{equation}
\hat{\psi}(y,A)=D_{1}(A)y^{B_{1}}+D_{2}(A)y^{B_{2}}+\frac{c-A}{r}y,
\label{dFBPiii}
\end{equation}
in which $B_1 > 1$ and $B_2 < 0$ are defined in (\ref{B}). Define 
\begin{equation}
y_{0}(A)=-\psi _{w}(0,A),  \label{y0iii}
\end{equation}
and 
\begin{equation}
y_{b}(A)=-\psi _{w}(w_{b}(A),A).  \label{ybiii}
\end{equation}
We get the following free-boundary conditions from (\ref{BCiii1}), (\ref{BCiii2}), (\ref{BCiii3}), (\ref{y0iii}), and (\ref{ybiii}):
\begin{equation}
\begin{cases}
\hat{\psi}_{A}(y_{0}(A),A)=-(1-p)\bar{a}y_{0}(A), \\ 
\hat{\psi}_{y}(y_{0}(A),A)=0;
\end{cases}
\label{sellbdryiii}
\end{equation}
\begin{equation}
\begin{cases}
\hat{\psi}_A(y_{b}(A),A) = -\bar{a}y_{b}(A), \\ 
\hat{\psi}_{y}(y_{b}(A),A) = w_{b}(A);
\end{cases}
\label{buybdryiii}
\end{equation}
and
\begin{equation}
\hat{\psi}(y_{0}(0),0)=1.  \label{ruinbdryiii}
\end{equation}
The smooth fit condition on the boundary $w=w_{b}(A)$ implies
\begin{equation}
\hat{\psi}_{Ay}(y_{b}(A),A)=-\bar{a}.  \label{dsmft}
\end{equation}

Use \eqref{dFBPiii} to rewrite (\ref{sellbdryiii}), (\ref{buybdryiii}), (\ref{ruinbdryiii}), and (\ref{dsmft}) as follows:
\begin{eqnarray}
D_{1}(A)B_{1}y_{0}(A)^{B_{1}-1}+D_{2}(A)B_{2}y_{0}(A)^{B_{2}-1}+{\frac{c-A}{r}} &=&0,
\label{iii1} \\
D_{1}^{\prime }(A)y_{0}(A)^{B_{1}-1}+D_{2}^{\prime }(A)y_{0}(A)^{B_{2}-1} &=& {\frac{1}{r}}-{\frac{1-p}{{r+\lambda ^{O}}}},  \label{iii2} \\
D_{1}(A)B_{1}y_{b}(A)^{B_{1}-1}+D_{2}(A)B_{2}y_{b}(A)^{B_{2}-1}+{\frac{c-A}{r}} &=&w_{b}(A),
\label{iii3} \\
D_{1}^{\prime }(A)y_{b}(A)^{B_{1}-1}+D_{2}^{\prime }(A)y_{b}(A)^{B_{2}-1} &=&{\frac{1}{r}}-{\frac{1}{{r+\lambda ^{O}}}},  \label{iii4} \\
D_{1}(0)y_{0}(0)^{B_{1}}+D_{2}(0)y_{0}(0)^{B_{2}}+{\frac{c}{r}}y_0(0) &=&1,
\label{iii6} \\
D_{1}^{\prime }(A) B_{1} y_{b}(A)^{B_{1}-1} + D_{2}^{\prime}(A) B_{2} y_{b}(A)^{B_{2}-1} &=&{\frac{1}{r}}-{\frac{1}{{r+\lambda ^{O}}}}.
\label{iii5}
\end{eqnarray}
Solve (\ref{iii1}) and (\ref{iii3}) for $D_{1}(A)$ and $D_{2}(A)$: 
\begin{eqnarray}
D_{1}(A) &=&\frac{1}{B_{1}}y_{b}(A)^{1-B_{1}}\frac{1}{x^{B_{1}-B_{2}}-1}\left[ -w_{b}(A)+\frac{c-A}{r} \left(1-x^{1-B_{2}} \right)\right] ,  \label{iiid1} \\
D_{2}(A) &=&\frac{1}{B_{2}}y_{b}(A)^{1-B_{2}}\frac{1}{x^{B_{2}-B_{1}}-1}\left[ -w_{b}(A)+\frac{c-A}{r} \left(1-x^{1-B_{1}} \right)\right] ,  \label{iiid2}
\end{eqnarray}
in which 
\begin{equation}
x\triangleq \frac{y_{0}(A)}{y_{b}(A)}.  \label{iiix1}
\end{equation}
Recall that $w_{b}(A)$ is to be determined. We solve for $D_{1}^{\prime }(A)$ and $D_{2}^{\prime }(A)$ from (\ref{iii4}) and (\ref{iii5}) to get: 
\begin{eqnarray}
D_{1}^{\prime }(A) &=&\frac{\lambda ^{O}}{r(r+\lambda ^{O})} \, \frac{1-B_{2}}{B_{1}-B_{2}} \, y_{b}(A)^{1-B_{1}},
\label{D1'} \\
D_{2}^{\prime }(A) &=&\frac{\lambda ^{O}}{r(r+\lambda ^{O})} \, \frac{B_{1}-1}{B_{1}-B_{2}} \, y_{b}(A)^{1-B_{2}}.
\label{D2'}
\end{eqnarray}
By substituting (\ref{D1'}) and (\ref{D2'}) into (\ref{iii2}), we get 
\begin{equation}
\frac{1-B_{2}}{B_{1}-B_{2}} \, \frac{\lambda ^{O}}{r(r+\lambda ^{O})} \, x^{B_{1}-1} + \frac{B_{1}-1}{B_{1}-B_{2}} \, \frac{\lambda ^{O}}{r(r+\lambda ^{O})} \, x^{B_{2}-1} = \frac{1}{r}-\frac{1-p}{r+\lambda ^{O}},
\label{iiix}
\end{equation}
which has a unique solution for $x>1$; the argument is similar to the corresponding one in Section \ref{sec:linearization} for the solution of \eqref{eq:**}.  It is clear from \eqref{iiix} that $x$ is independent of $A$.

Differentiate $D_{1}(A)$ and $D_{2}(A)$ in (\ref{iiid1}) and (\ref{iiid2}) with respect to $A$ to get a second expression for $D_{1}^{\prime }(A)$ and $D_{2}^{\prime }(A)$; set equal the two expressions for each of $D_{1}^{\prime }(A)$ and $D_{2}^{\prime }(A)$ to get 
\begin{eqnarray}
\frac{dy_{b}(A)/dA}{y_{b}(A)} &=& \frac{\dfrac{\lambda ^{O}}{r(r+\lambda ^{O})} \, \dfrac{B_{1}(1-B_{2})}{B_{1}-B_{2}} \left(x^{B_{1}-B_{2}}-1 \right) + w_{b}^{\prime }(A) + \dfrac{1}{r} \left(1-x^{1-B_{2}} \right)}{(1-B_{1}) \left[-w_{b}(A)+\dfrac{c-A}{r} \left(1-x^{1-B_{2}}\right) \right]},
\label{dyiii} \\
\frac{dy_{b}(A)/dA}{y_{b}(A)} &=& \frac{\dfrac{\lambda ^{O}}{r(r+\lambda ^{O})} \, \dfrac{B_{2}(B_{1}-1)}{B_{1}-B_{2}} \left(x^{B_{2}-B_{1}}-1 \right) + w_{b}^{\prime }(A) + \dfrac{1}{r} \left(1-x^{1-B_{1}} \right)}{(1-B_{2}) \left[-w_{b}(A)+\dfrac{c-A}{r} \left(1-x^{1-B_{1}} \right) \right]}.
\end{eqnarray}
Set equal the right-hand sides of the two equations above to get a non-linear ODE for $w_{b}(A)$: 
\begin{equation}
\alpha _{1}(c-A)w_{b}^{\prime }(A)+\alpha _{2}w_{b}(A) +\alpha _{3}w_{b}^{\prime}(A)w_{b}(A)+\alpha _{4}(c-A)=0,
\label{eq:w_b}
\end{equation}
in which 
\begin{equation}
\begin{cases}
\alpha_{1} = - \dfrac{1}{r}\left[(B_{1} - 1)\left(1-x^{1-B_{2}}\right) + (1-B_{2})\left(1-x^{1-B_{1}}\right)\right] , \\ 
\alpha_{2} =  (B_{1} - 1) \left[ \dfrac{\lambda ^{O}}{r(r+\lambda ^{O})}\dfrac{(B_{1}-1)B_{2}}{B_{1}-B_{2}}\left(x^{B_{2}-B_{1}}-1\right)+\dfrac{1}{r}\left(1-x^{1-B_{1}}\right) \right] \\
\qquad \; \, + (1-B_{2})\left[ \dfrac{\lambda ^{O}}{r(r+\lambda ^{O})}\dfrac{B_1 (1-B_{2})}{B_{1}-B_{2}}\left(x^{B_{1}-B_{2}}-1\right)+\dfrac{1}{r}\left(1-x^{1-B_{2}}\right)
\right], \\ 
\alpha_{3} = B_{1}-B_{2} > 0, \\ 
\alpha_{4} = - \dfrac{1}{r}\left\{(B_{1} - 1)\left[ \dfrac{\lambda ^{O}}{r(r+\lambda^{O})}\dfrac{(B_{1}-1) B_2}{B_{1}-B_{2}}\left(x^{B_{2}-B_{1}}-1\right)+\dfrac{1}{r} \left(1-x^{1-B_{1}}\right)\right] \left(1-x^{1-B_{2}}\right)\right. \\ 
\left. \qquad \qquad + (1-B_{2})\left[ \dfrac{\lambda ^{O}}{r(r+\lambda ^{O})}\dfrac{B_1 (1-B_{2})}{B_{1}-B_{2}}\left(x^{B_{1}-B_{2}}-1\right)+\dfrac{1}{r}\left(1-x^{1-B_{2}}\right) \right] \left(1-x^{1-B_{1}}\right)\right \}.
\end{cases}
\label{eq:alpha}
\end{equation}
Also, we have the boundary condition $w_b(c-) = 0$ because $0 \le w_b(A) \le w_s(A)$ for all $0 \le A < c$ and $w_s(c-) = 0$.  A solution of the ODE, together with the boundary condition at $A = c$, is given by
\begin{equation}
w_{b}(A) = b \cdot (c-A),  \label{wb}
\end{equation}
in which 
\begin{equation}
b={\frac{(\alpha _{2}-\alpha _{1})+\sqrt{(\alpha _{2}-\alpha_{1})^{2}+4\alpha _{3}\alpha _{4}}}{{2\alpha _{3}}}}.  \label{eq:b}
\end{equation}
Note that this solution for the purchase boundary $w_{b}(A)$ is linear with respect to $A$.

From the expression on the right-hand side of \eqref{dyiii} and from \eqref{wb}, define
\begin{equation}
K\triangleq \dfrac{\dfrac{\lambda ^{O}}{r(r+\lambda ^{O})} \, \dfrac{B_{1}(1-B_{2})}{B_{1}-B_{2}} \left(x^{B_{1}-B_{2}}-1 \right) - b + \dfrac{1}{r} \left(1-x^{1-B_{2}} \right)}{(1-B_{1}) \left[-b + \dfrac{1}{r} \left(1-x^{1-B_{2}} \right) \right]}.  \label{iiiK}
\end{equation}
Solve (\ref{dyiii}) and (\ref{iiix1}) to obtain
\begin{equation}
y_{0}(A) = \left(\frac{c}{c-A} \right)^{K}y_{0}(0),  \label{iiiy0}
\end{equation}
and 
\begin{equation}
y_{b}(A)=\frac{y_{0}(A)\text{ }}{x}.  \label{iiiys}
\end{equation}

To finish solving the FBP, we substitue (\ref{iiid1}), (\ref{iiid2}), and (\ref{wb}) into (\ref{iii6}) to get
\begin{equation}
\frac{1}{y_{0}(0)}=\frac{c}{B_{1}}\frac{x^{B_{1}-1}}{x^{B_{1}-B_{2}}-1}\left[-b + \frac{1}{r} \left(1-x^{1-B_{2}} \right)\right] +\frac{c}{B_{2}}\frac{x^{B_{2}-1}}{x^{B_2-B_{1}}-1}\left[ -b + \frac{1}{r} \left(1-x^{1-B_{1}} \right)\right] + \frac{c}{r}.
\label{iiiy00}
\end{equation}

\begin{proposition}
\label{prop:psihatiii}The solution of the FBP \eqref{FBPiii} with conditions \eqref{sellbdryiii}, \eqref{buybdryiii}, and \eqref{ruinbdryiii} is given by \eqref{dFBPiii}, with $D_{1}(A)$, $D_{2}(A)$, $y_{0}(0)$, $y_{0}(A)$, $y_{b}(A)$, $x$, and $K$ defined in \eqref{iiid1}, \eqref{iiid2}, \eqref{iiiy00}, \eqref{iiiy0}, \eqref{iiiys}, \eqref{iiix}, and \eqref{iiiK}, respectively.
\end{proposition}

Notice that we can rewrite the inequalities (\ref{ineqi1}) and (\ref{ineqi2}) in terms of $\hat{\psi}$ as 
\begin{eqnarray}
\hat{\psi}_{A}(y,A) &\geq &-\frac{1}{r+\lambda ^{O}}y,  \label{dineqiii1} \\
\hat{\psi}_{A}(y,A) &\leq &-\frac{1-p}{r+\lambda ^{O}}y.  \label{dineqiii2}
\end{eqnarray}
Next, we prove that $\hat \psi$ is concave with respect to $y$ and satisfies inequalities (\ref{dineqiii1}) and (\ref{dineqiii2}).

\begin{proposition}
\label{prop:psihat_concaveiii}
$\hat{\psi }(y,A)$ given by Proposition \ref{prop:psihatiii} is concave with respect to $y$ and satisfies inequalities \eqref{dineqiii1} and \eqref{dineqiii2}.
\end{proposition}

\begin{proof}
The proof that $\hat{\psi}$ is concave with respect to $y$ is not obvious (unlike the previous two cases), so we relegate that (long) proof to the Appendix.

Substitute $D'_{1}(A)y^{ B_{1}}+D'_{2}(A)y^{B_{2}}-\dfrac{y}{r}$ for $%
\hat{\psi}_{A}(y,A)$ to rewrite the inequalities (\ref{dineqiii1}) and (\ref{dineqiii2}) in the equivalent form as 
\begin{equation}
-1\leq \frac{\lambda ^{O}}{r}\left[ \frac{1-B_{2}}{B_{1}-B_{2}}\left( \frac{y%
}{y_{b}(A)}\right) ^{B_{1}-1}+\frac{B_{1}-1}{B_{1}-B_{2}}\left( \frac{y}{%
y_{b}(A)}\right) ^{B_{2}-1}\right] -\frac{r+\lambda ^{O}}{r}\leq -(1-p).
\label{iiiverf}
\end{equation}%
To prove the inequality above, define 
\begin{equation}
h(z)=\frac{1-B_{2}}{B_{1}-B_{2}}z^{B_{1}-1}+\frac{B_{1}-1}{B_{1}-B_{2}}z^{B_{2}-1},
\end{equation}%
and note that 
\begin{equation}
h^{\prime }(z)=\frac{(1-B_{2})(B_{1}-1)}{B_{1}-B_{2}}\left[ z^{B_{1}-2}-z^{B_{2}-2}\right] \geq 0, \quad z \ge 1.
\end{equation}%
Also, the first inequality in \eqref{iiiverf} holds with equality when $y = y_b(A)$, and the second inequality holds with equality when $y = y_0(A)$.  Thus, because $h(z)$ is non-decreasing for $z \ge 1$, inequality \eqref{iiiverf} holds for $y_{b}(A)\leqslant y\leqslant y_{0}(A)$.
\end{proof}

As before, we define the convex dual of $\hat{\psi}$ via the Legendre transform for $(w,A) \in \mathcal{D}_1$ as 
\begin{equation}
\Psi (w,A)=\max_{y\geq y_b(A)} \left[\hat{\psi}(y,A)-wy \right].  \label{5.57}
\end{equation}
For $(w,A) \in \mathcal{D}_2$, we define 
\begin{equation}
\Psi (w,A)=\Psi (w-\bar{a}\Delta A,A+\Delta A),  \label{defiii}
\end{equation}%
in which $\Delta A$ solves $w-\bar{a}\Delta A=b(c-(A+\Delta A))$; that is, $\Delta A=\dfrac{w-b(c-A)}{\bar{a}-b}$. Notice that
since $(w-\bar{a}\Delta A,A+\Delta A) \in \mathcal{D}_1$, $\Psi (w-\bar{a}\Delta A,A+\Delta A)$ is given through (\ref{5.57}). 

Now we proceed to the following lemma, which demonstrates that $\Psi$ is the minimum probability of ruin by Verification Theorem \ref{lem:verf-lemma}.
\begin{lemma}
\label{lem:4.1}  $\Psi (w,A)$ defined in \eqref{5.57} and \eqref{defiii} satisfies Conditions $1$-$4$ of the Verification Theorem \ref{lem:verf-lemma}.
\end{lemma}

\begin{proof}
First, consider $(w,A) \in \mathcal{D}_1$. In terms of $\Psi(w,A)$, we rewrite (\ref{FBPiii}) as follows:
\begin{equation}
\lambda ^{S}\Psi (w,A)=(rw-c)\Psi _{w}(w,A)-m\frac{\Psi _{w}^{2}(w,A)}{\Psi_{ww}(w,A)},  \label{5.58}
\end{equation}
as well as  (\ref{dineqiii1}) and (\ref{dineqiii2}) 
\begin{eqnarray}
\bar{a}\Psi _{w}(w,A)&\leq &\Psi _{A}(w,A) ,  \label{condition2} \\
(1-p)\bar{a}\Psi _{w}(w,A)&\geq &\Psi _{A}(w,A)  \label{condition3}.
\end{eqnarray} 
Expressions (\ref{5.58})-(\ref{condition3}) show that $\Psi(w,A)$ satisfies Conditions $1$-$3$ of the Verification Theorem \ref{lem:verf-lemma}  on $\mathcal{D}_1$.  It is clear by construction that $\Psi$ satisfies Condition 4, namely, $\Psi(0, 0) = 1$.

Now, consider $(w,A) \in \mathcal{D}_2$. By definition, 
\begin{equation}
\Psi (w,A)=\Psi (w^{\prime },A^{\prime }),
\label{eq:496}
\end{equation}
with $w'=w-\dfrac{w-b(c-A)}{1-b/\bar{a}}$ and $A^{\prime }=A+\dfrac{w-b(c-A)}{\bar{a}-b}$.  From \eqref{eq:496}, we get the following relations
\begin{equation}
\Psi _{w}(w,A) = - \, \frac{b}{\bar{a}-b} \, \Psi _{w}(w^{\prime },A^{\prime })  + \frac{1}{\bar{a}-b} \, \Psi _{A}(w^{\prime },A^{\prime }) ,
\label{1}
\end{equation}
\begin{equation}
\Psi _{ww}(w,A) = \left( \frac{b}{\bar{a}-b}\right)^{2} \, \Psi _{ww}(w^{\prime },A^{\prime }) - \frac{2b}{(\bar{a}-b)^{2}} \, \Psi _{wA}(w,A) + \left( \frac{1}{\bar{a}-b} \right)^{2} \, \Psi _{AA}(w^{\prime },A^{\prime }),
\label{2}
\end{equation}
and
\begin{equation}
\Psi _{A}(w,A) = - \, \frac{b \bar a}{\bar a - b} \, \Psi _{w}(w^{\prime },A^{\prime }) + \frac{\bar a}{\bar a - b}  \, \Psi _{A}(w^{\prime },A^{\prime }).
\label{3}
\end{equation}

Since $(w^{\prime },A^{\prime })$ is on the boundary $w=w_{b}(A)$, we have 
\begin{equation}
\bar{a}\Psi _{w}(w^{\prime },A^{\prime })=\Psi _{A}(w^{\prime },A^{\prime }).
\label{4} 
\end{equation}
This along with (\ref{1}) leads to 
\begin{equation}
\Psi _{w}(w,A)=\Psi _{w}(w^{\prime },A^{\prime }).
\label{44}
\end{equation}
Differentiate (\ref{44}) with respect to $w$ to get 
\begin{equation}
\Psi _{ww}(w,A) = - \, \frac{b}{\bar{a}-b} \, \Psi _{ww}(w^{\prime },A^{\prime })  + \frac{1}{\bar{a}-b} \, \Psi _{wA}(w^{\prime },A^{\prime }).
\label{eq:4.102}
\end{equation}
From \eqref{eq:4.102} and from the smooth fit condition on the purchase boundary, namely 
\begin{equation}
\bar{a}\Psi _{ww(}w^{\prime },A^{\prime })=\Psi _{wA}(w^{\prime },A^{\prime}),
\end{equation}
we  obtain 
\begin{equation}
\Psi _{ww}(w,A)=\Psi_{ww} (w^{\prime },A^{\prime }).
\end{equation}

We know that for $(w',A') \in \mathcal{D}_1$ and for $\pi \in \mathbb{R}$,
\begin{equation}
{\cal L}^{\pi }\Psi(w',A') =[rw'+(\mu -r)\pi -c+A']\Psi_{w}(w',A')+\frac{1}{2}\sigma ^{2}\pi ^{2}\Psi
_{ww}(w',A')-\lambda ^{S}\Psi(w',A') \geq 0.
\end{equation}
It follows that for $(w,A) \in \mathcal{D}_2$, 
\begin{equation}
\begin{split}
\mathcal{L}^{\pi} \Psi (w,A)=& \left[ r\left( w^{\prime }+\frac{w-b(c-A)}{1-b/\bar{a}} \right) +(\mu -r)\pi -c+\left( A^{\prime }-\frac{w-b(c-A)}{\bar{a}-b}\right) \right] \Psi_{w}(w^{\prime },A^{\prime }) \\
&  + \frac{1}{2} \sigma^{2} \pi^{2}\Psi _{ww}(w^{\prime },A^{\prime}) - \ls \, \Psi (w^{\prime },A^{\prime }) \\
=& \mathcal{L}^{\pi }\Psi (w^{\prime },A^{\prime })+\left[ r\frac{w-b(c-A)}{1-b/\bar{a}
}-\frac{w-b(c-A)}{\bar{a}-b}\right] \Psi _{w}(w^{\prime },A^{\prime })\geq 0,
\end{split}
\end{equation}
because $\mathcal{L}^{\pi }\Psi (w^{\prime },A^{\prime })\geq 0$, $\Psi _{w}(w^{\prime
},A^{\prime })\leq 0$, and $\dd{r \, \frac{w-b(c-A)}{1-b/\bar{a}}-\frac{w-b(c-A)}{
\bar{a}-b}\leq 0}$.  Thus, $\Psi(w,A)$ satisfies Condition $1$ of the Verification Theorem for $(w,A) \in \mathcal{D}_2$. 

Next, consider Conditions 2 and 3 for $(w,A) \in \mathcal{D}_2$.  From \eqref{1} and \eqref{3}, Condition 2 holds for $(w, A)$ if and only if
\begin{equation}
\bar a \left[ - \, \frac{b}{\bar{a}-b} \, \Psi _{w}(w^{\prime },A^{\prime })  + \frac{1}{\bar{a}-b} \, \Psi _{A}(w^{\prime },A^{\prime }) \right] \le - \, \frac{b \bar a}{\bar a - b} \, \Psi _{w}(w^{\prime },A^{\prime }) + \frac{\bar a}{\bar a - b}  \, \Psi _{A}(w^{\prime },A^{\prime }),
\end{equation}
which is true with equality.  Thus, we conclude that Condition 2 holds with equality for $(w, A) \in \mathcal{D}_2$.  Finally, because $0 < p < p^* \le 1$, $\Psi_w \le 0$, and $\Psi_A \le 0$, it follows that Condition 3 also holds on $\mathcal{D}_2$.
\end{proof}

Therefore, $\Psi (w,A)$ is the minimum probability of ruin by the Verification Theorem \ref{lem:verf-lemma}, and we present the following theorem that summarizes the work of this section.

\begin{thm}
When $p<p^{\ast }$ and the borrowing restriction is enforced, the minimum probability of ruin for $(w,A)\in \mathcal{D}=\mathcal{D}_{1}\cup \mathcal{D}_{2}$, with $\mathcal{D}_{1}=\{(w,A):0\leq w\leq w_{b}(A),0\leq A<c\}$ and $\mathcal{D}%
_{2}=\{(w,A):w_{b}(A)<w<w_{s}(A),0\leq A<c\}$, is given by $\Psi (w,A)$
defined above. The associated optimal strategy is:

\begin{enumerate}
\item to purchase additional annuity income so that wealth and annuity income lie on the boundary $w = b \cdot (c - A)$ of the region $\mathcal{D}_{1}$ when $(w,A)\in \mathcal{D}_{2}$;

\item to purchase additional annuity income instantaneously to keep $(w,A)$ in the region $\mathcal{D}_{1}$ when $w=w_{b}(A)$;

\item to surrender exisiting annuity income instantaneously to keep $w$ non-negative when needed;

\item to invest in the risky asset with amount 
\begin{equation*}
\pi ^*(w,A)=-\frac{\mu -r}{\sigma ^{2}}\frac{\Psi _{w}(w,A)}{\Psi _{ww}(w,A)%
},
\end{equation*}%
when $(w,A)\in \mathcal{D}_{1}$.
\end{enumerate}
\end{thm}

\subsection{Numerical Examples}

In this section, we present numerical examples to demonstrate the results of Section \ref{sec:p-big} and \ref{sec:p-small}. The basic scenario is the same as in Section \ref{sec:num1}, and we focus on the role of the surrender penalty $p$.

\subsubsection{$p \geq p^*$}

Figures \ref{fig:2.1}-\ref{fig:2.4} show the ruin probabilities and associated optimal investment strategies when $p \geq p^*$.  We fix all the parameters except for annuity income $A$ and surrender penalty $p$.  Note that the smallest $p$ value of the selected is $0.258$, which is the value of $p^*$ for the scenario we chose.  The boundary $w=0$ and $w_s$ does not depend on $p$. Therefore, for each figure, all four curves have the same domain.  By showing the ruin probabilities and investment strategies for different $A$ and $p$, we see some common patterns as well as the effect of $p$.  Within each figure, the probabilities of ruin are decreasing and convex.  On the surrender boundary $w=0$, the curves of the ruin probabilities begin with different values, not necessarily $0$.  This occurs because if the individual has some annuity income, she surrenders some of it to avoid ruin when reaching that boundary.  We also observe that bigger $p$ results in higher probability of ruin. This is consistent with the financial intuition that an individual receives less wealth from surrendering annuity income when the penalty $p$ is bigger, as we also show in Proposition \ref{prop:iipi}.   

If $A$ is not $0$, reversibility makes difference in the ruin probability on the boundary $w=0$, and consequently on the whole ruin probability curve.  Reversibility of the annuity offers an extra chance to avoid bankrupcty.  This is demonstrated by the difference of ruin probabilities between $p=0.258$ and $p=1$ for given values of $(w, A)$.  Note that at $w=0$, the difference increases dramatically as $A$ increases.  When $A=0$, both ruin probabilities are $1$.  On the other hand, when $A=0.75$, the individual with the reversible annuity ($p = 0.258$) has only about a $25\%$ chance to ruin when her wealth is $0$ if she follows the optimal strategy.  By contrast, if the annuity is irreversible ($p = 1$), she ruins immediately when wealth is $0$ because the annuity is effectively worthless.  This gap in the ruin probabilities shrinks as wealth $w$ increases, and  the ruin probabilities associated with different $p$'s  converge to $0$ at $w=w_s(A)$, the safe level. 

The interesting phenomenon that the investment in the risky asset does not depend on $p$ is demonstrated in all figures, as we also show in Proposition \ref{prop:iipi}. This indicates that the individual invests in the risky asset as if the annuity is irreversible when $p \ge p^*$.  That is, we see a type of separation result:  optimal investment in the risky asset is independent of the optimal annuitization strategy when $p \ge p^*$.

\subsubsection{$p<p^*$}

Figures  \ref{fig:3.1}-\ref{fig:3.4} show the ruin probabilities and associated optimal investment strategies when $p<p^*$.  Recall from Section \ref{sec:p-small} that it is optimal to purchase immediate life annuities before wealth reaches the safe level.  Note that the largest value of $p$ we can choose is $0.258$.  It is natural to believe that one's behavior changes smoothly as penalty $p$ changes. This belief is confirmed in these figures. By observing the curves associated with  $p=0.258$ in Figures \ref{fig:2.1}-\ref{fig:2.4} and in Figures \ref{fig:3.1}-\ref{fig:3.4}, we conclude that the optimal investment strategies and ruin probabilities from the two different sets of equations are the same.  (We can also demonstrate this fact algebraically, but in the interest of space, we omit that computation.)

We see that the ruin probabilities in Figures \ref{fig:3.1}-\ref{fig:3.4} are all decreasing and convex.  The wealth domain for a given function in these figures is $[0, b \cdot (c - A)]$, and note that $b$ decreases as $p$ decreases because for a smaller surrender charge, the individual has more incentive to annuitize at a lower wealth level.   It remains true that, with all else equal, a smaller surrender charge $p$ results in a smaller probability of ruin.  Also, note that investment in risky asset increases as wealth increases, as in the case for which $p \ge p^*$.  However, different from what we see for $p \geq p^*$ case, the investment strategy is no longer independent of $p$.   More cash is invested in the risky asset if one can get a larger portion of her annuity value back by surrendering.

Figure \ref{fig:bp} demonstrates the relation between $b$ and the proportional surrender penalty $p$. The sign $*$ in the figure indicates the $b$ value of $\dfrac{1}{r+\lo}$. We see that $b$ increases monotonically and continuously from $0$ to $\dfrac{1}{r+\lo}$ as $p$ increases from $0$ to $p^*$, as we expect. 

\section{Conclusion}

The annuity puzzle has been widely noted both in practice and in theoretical work; see \cite{Milevsky_Young} and \cite{Milevsky_Moore_Young} and the references therein.  In this paper, we considered a financial innovation that might encourage more retirees to purchase immediate life annuities, namely the option to surrender one's annuity for cash.  We explained the relation between the irreversibility of annuitization and the retirees' reluctance to purchase.  We investigated how reversibility would affect the decision of a retiree seeking to minimize her lifetime probability of ruin.  We analyzed the optimal investment and annuitization strategies for such a retiree when borrowing against the surrender value of the annuity is prohibited.  We found that the individualÕs annuity purchasing strategy depends on the size of the proportional surrender charge.  When the charge is large enough, the individual will not buy a life annuity unless she can cover all her consumption, the so-called safe level.  When the charge is small enough, the individual will buy a life annuity at a wealth lower than this safe level.  In both cases, the individual only surrenders annuity income in order to keep her wealth non-negative.

These results confirm the point of view in \cite{Gardner_Wadsworth} that the lack of flexibility discourages retirees from purchasing immediate life annuities.  In our model, if annuities are irreversible, then retirees will buy annuities only when their wealth reaches the safe level.  Moreover, we showed that if annuities are reversible, then a retiree will partially annuitize if the surrender charge is low enough.  In numerical examples, we noticed that the threshold value of surrender charge for an individual to consider partial annuitization might be too low for annuity providers.  This perhaps explains why  reversible immediate life annuities are not offered in the annuity market.

The model in this paper offers a mathematical framework to understand the annuity puzzle.  Even though we assumed constant hazard rates and interest rate in our analysis, we believe that the main qualitative insight will be true in general and will be useful to develop better structured annuity products for retirees.  Our analysis also implies that  a well developed secondary market of annuities would benefit both potential annuity buyers and providers.

\section{Appendix}

In this appendix, we prove that the $\hat{\psi}$ given in Proposition \ref{prop:psihatiii} is concave thereby completing the proof of Proposition \ref{prop:psihat_concaveiii}. 

Take the second derivative of (\ref{dFBPiii}) with respect to $y$ to get
\begin{equation}\label{app1}
\hat{\psi}_{yy}(y,A)=D_1(A)B_1(B_1-1)y^{B_1-2}+D_2(A)B_2(B_2-1)y^{B_2-2}. 
\end{equation}
We want to show that $\hat{\psi}_{yy}(y,A)\leq 0$ for $y_b(A)\leq  y \leq y_0(A)$. Substitute (\ref{iiid1}) and (\ref{iiid2}) into (\ref{app1}), and define ${z} \triangleq y/y_b(A) \in [1,x] $, with $x$ defined by (\ref{iiix1}). Then, we get 
\begin{equation}\label{app2}
\begin{split}
\hat{\psi}_{yy}(y,A) \leq 0 \iff &(B_1-1)\left[ \frac{y}{y_b(A)} \right]^{B_1-B_2} \dfrac{1}{x^{B_1-B_2}-1} \left[ -b + \frac{1}{r} \left( 1 - x^{1 - B_2} \right) \right] \\
& + (1-B_2)\dfrac{x^{B_1-B_2}}{x^{B_1-B_2}-1}\left[ -b + \frac{1}{r} \left( 1 - x^{1 - B_1} \right) \right] \leq 0\\
\iff & (B_1-1)\left[ -b+\frac{1}{r}\left(1-x^{1-B_2}\right) \right] \left(\dfrac{z}{x}\right)^{B_1-B_2} \\
& + (1-B_2) \left[ -b+\frac{1}{r}\left(1-x^{1-B_1} \right ) \right] \leq0. \\
\end{split}
\end{equation}\label{app3}
Note that $B_1-1>0$ and $x^{B_2-1} < 1$. It follows that
\begin{equation}
 (B_1-1)\left[ -b+\frac{1}{r}\left(1-x^{1-B_2}\right) \right]<0.
\end{equation}
Hence the left-hand side of the last inequality in (\ref{app2}) reaches its maximum value when $z = 1$. So, to prove that $\hat \psi$ is concave with respect to $y$, it is sufficient to show that
\begin{equation}\label{app4}
(B_1-1)\left[ -b+\frac{1}{r}\left(1-x^{1-B_2}\right) \right] x^{B_2 - 1} \\
+(1-B_2) \left[ -b+\frac{1}{r}\left(1-x^{1-B_1} \right ) \right] x^{B_1 - 1} \leq0. \\
\end{equation}
Solve for $x^{B_2-1}$ from (\ref{iiix}); then, substitute into (\ref{app4}), which becomes 
\begin{equation}\label{app5}
\begin{split}
& (B_1-1)\left[ \left(\frac{1}{r}-b\right)\frac{B_1-B_2}{B_1-1} \frac{\lo+pr}{\lo} - \left(\frac{1}{r}-b \right) \frac{1-B_2}{B_1-1} x^{B_1-1} -\frac{1}{r}\right]\\
&+(1-B_2) \left[\left(\frac{1}{r}-b\right) x^{B_1-1}-\frac{1}{r}\right]\leq0\\
\iff& \left(\frac{1}{r}-b\right)(B_1-B_2)\frac{\lo+pr}{\lo} -\frac{1}{r} (B_1-B_2)\leq0  \iff b\geq \frac{p}{\lo+pr}.
\end{split}
\end{equation}

Therefore, if we show that $b \ge p/(\lo + pr)$, then we are done.  To this end, note that
\begin{equation}
b\geq \frac{p}{\lo+pr} \iff \frac{2 \alpha_3}{r} \left(1 - \frac{\lo}{\lo + pr} \right) - (\alpha_2 - \alpha_1) \le \sqrt{ (\alpha_2 - \alpha_1)^2 + 4 \alpha_3 \alpha_4 },
\label{app9}
\end{equation}
in which the $\alpha_i$ are given in \eqref{eq:alpha} for $i = 1, \dots, 4$.  The second inequality above holds automatically if its left-hand side is less than or equal to $0$.  Thus, suppose that the left-hand side is positive, and square both sides to get that $b \ge p/(\lo + pr)$ holds if
\begin{equation}\label{app10}
\alpha_4 + {1 \over r} \left( 1 - {\lo \over \lo + pr} \right) (\alpha_2 - \alpha_1) - {\alpha_3 \over r} \left(1 -  {\lo \over \lo + pr}  \right)^2 \ge 0.
\end{equation}
By substituting for the $\alpha_i$, $i = 1, \dots, 4$, by substituting for $\lo/(\lo + pr)$ via the following expression from \eqref{iiix} 
\begin{equation}\label{app11}
\frac{\lo}{\lo+pr}=\dfrac{(B_1-B_2) x^{1 - B_1} x^{1 - B_2}}{(B_1 - 1)x^{1 - B_1} + (1 - B_2)x^{1 - B_2}},
\end{equation}
and by simplifying carefully, we learn that \eqref{app10} is equivalent to
\begin{equation}\label{app13}
0 \leq \frac{\lo}{\lo + pr} - \frac{\lo}{\lo + r},
\end{equation}
which is true because $0<p  \le 1$.  We have proved that $b \ge p/(\lo+pr)$ and, thereby, that $\hat{\psi}$ is concave with respect to $y$.

\newpage

\begin{figure}
	\centering
		\includegraphics[angle=90, width=1.00\textwidth, height=0.9\textheight]{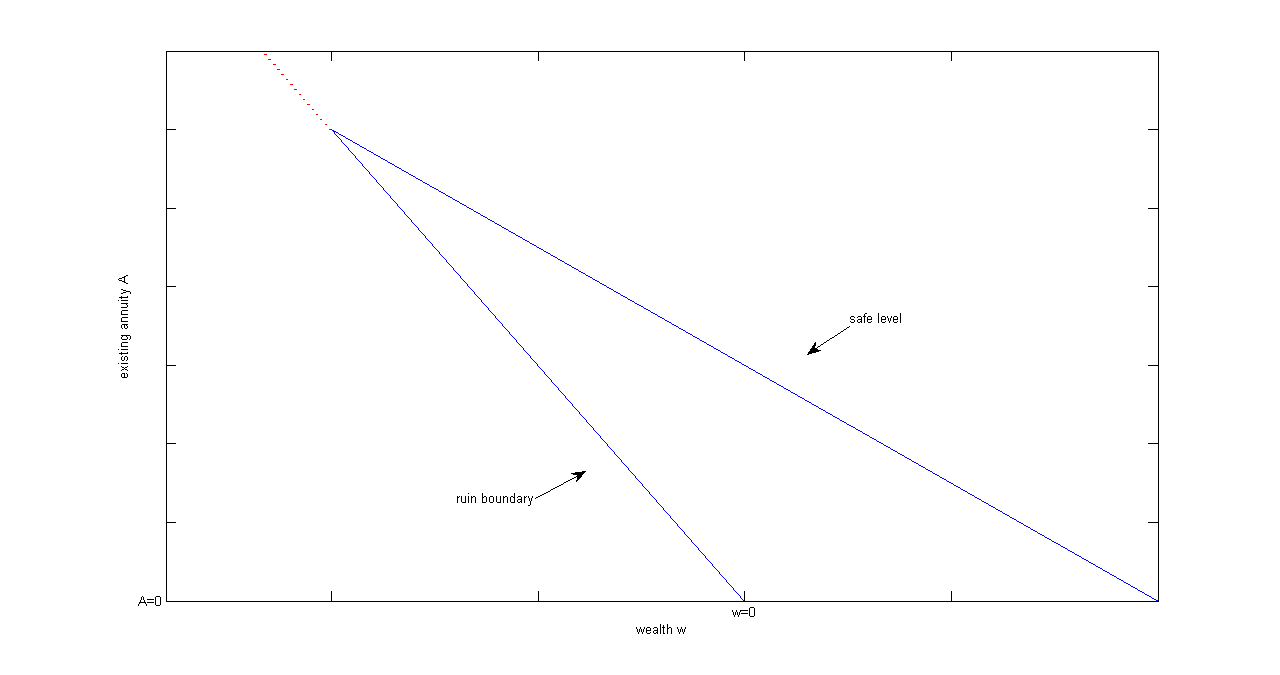}
	\caption{The region for solving minimum probability of ruin when borrowing against annuity is allowed}
	\label{fig:region}
\end{figure}

\begin{figure}
	\centering
		\includegraphics[angle=90, width=1.00\textwidth, height=0.9\textheight]{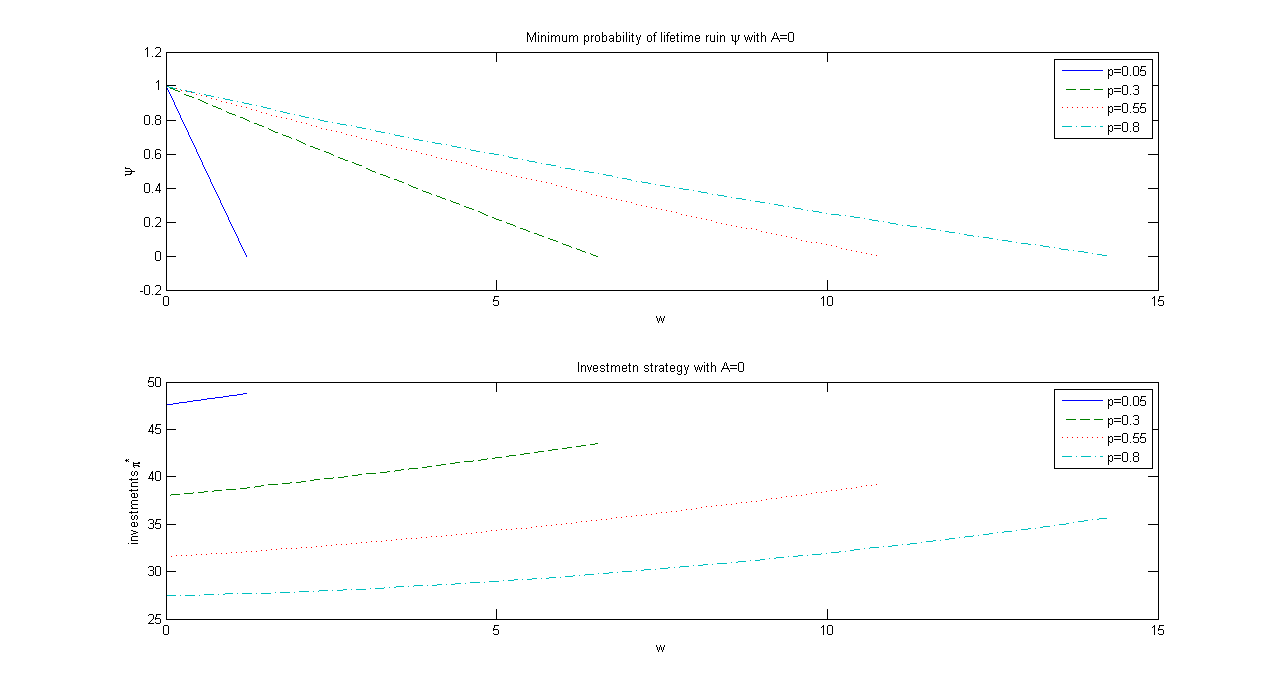}
	\caption{Ruin probabilities and optimal investment strategies for different $p$ when $A$ is $0$}
	\label{fig:1.1}
\end{figure}

\begin{figure}
	\centering
		\includegraphics[angle=90, width=1.00\textwidth, height=0.9\textheight]{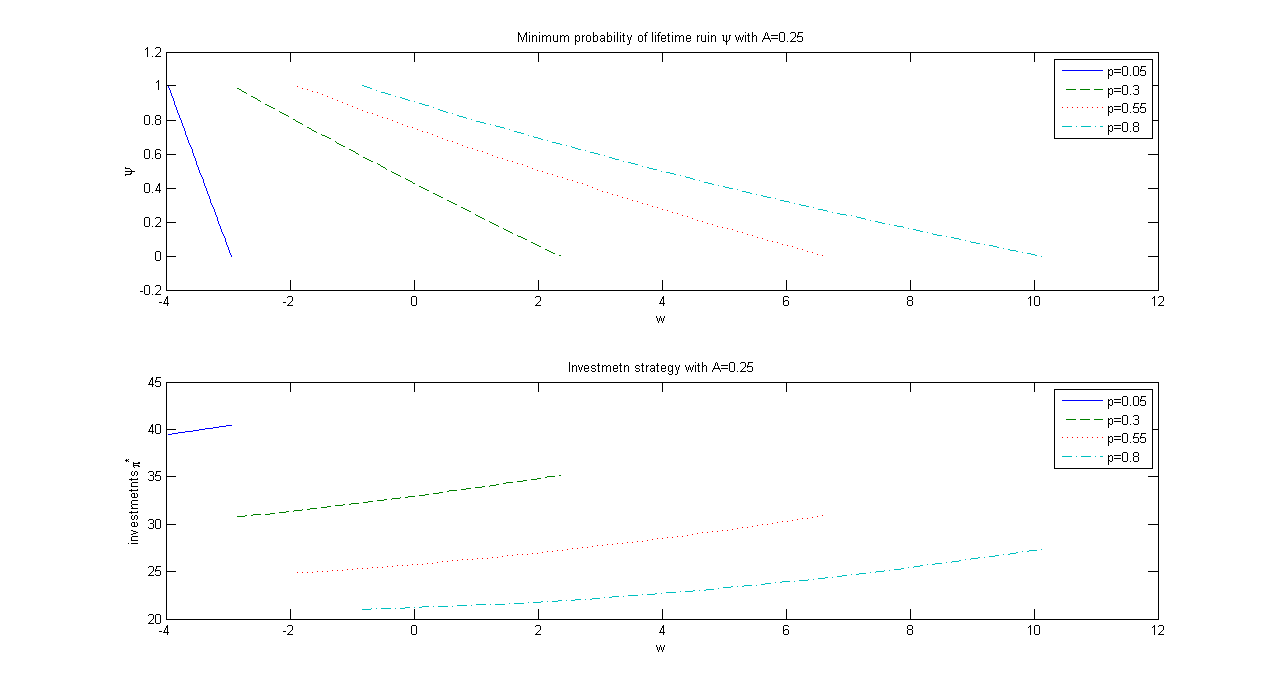}
	\caption{Ruin probabilities and optimal investment strategies for different $p$ when $A$ is $0.25$}
\label{fig:1.2}
\end{figure}

\begin{figure}
	\centering
		\includegraphics[angle=90, width=1.00\textwidth, height=0.9\textheight]{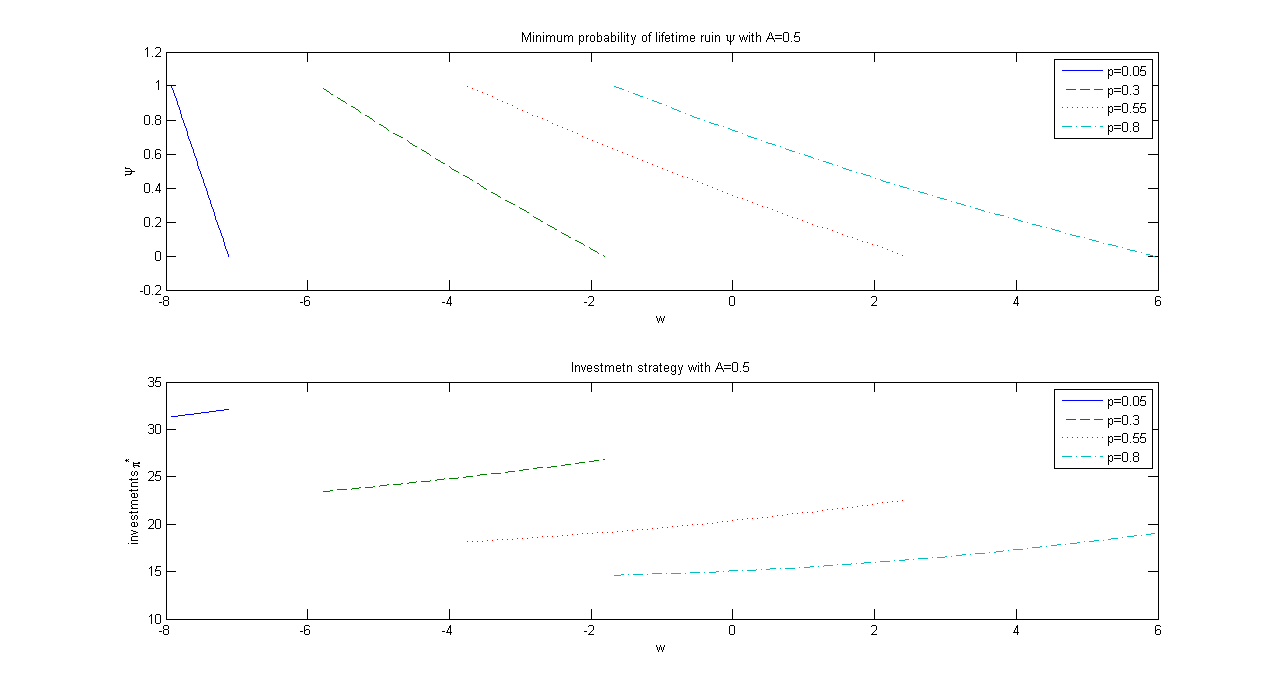}
	\caption{Ruin probabilities and optimal investment strategies for different $p$ when $A$ is $0.5$}
	\label{fig:1.3}
\end{figure}

\begin{figure}
	\centering
		\includegraphics[angle=90, width=1.00\textwidth, height=0.9\textheight]{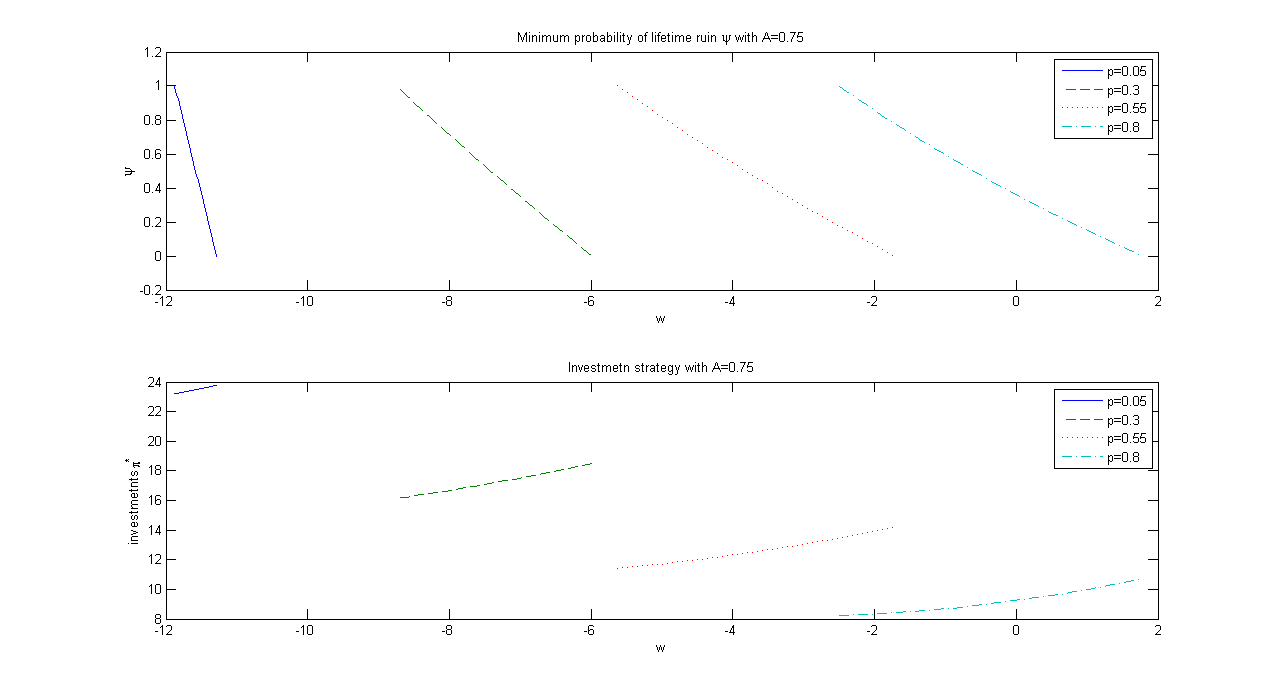}
	\caption{Ruin probabilities and optimal investment strategies for different $p$ when $A$ is $0.75$}
	\label{fig:1.4}
\end{figure}

\begin{figure}
	\centering
		\includegraphics[angle=90, width=1.00\textwidth, height=0.9\textheight]{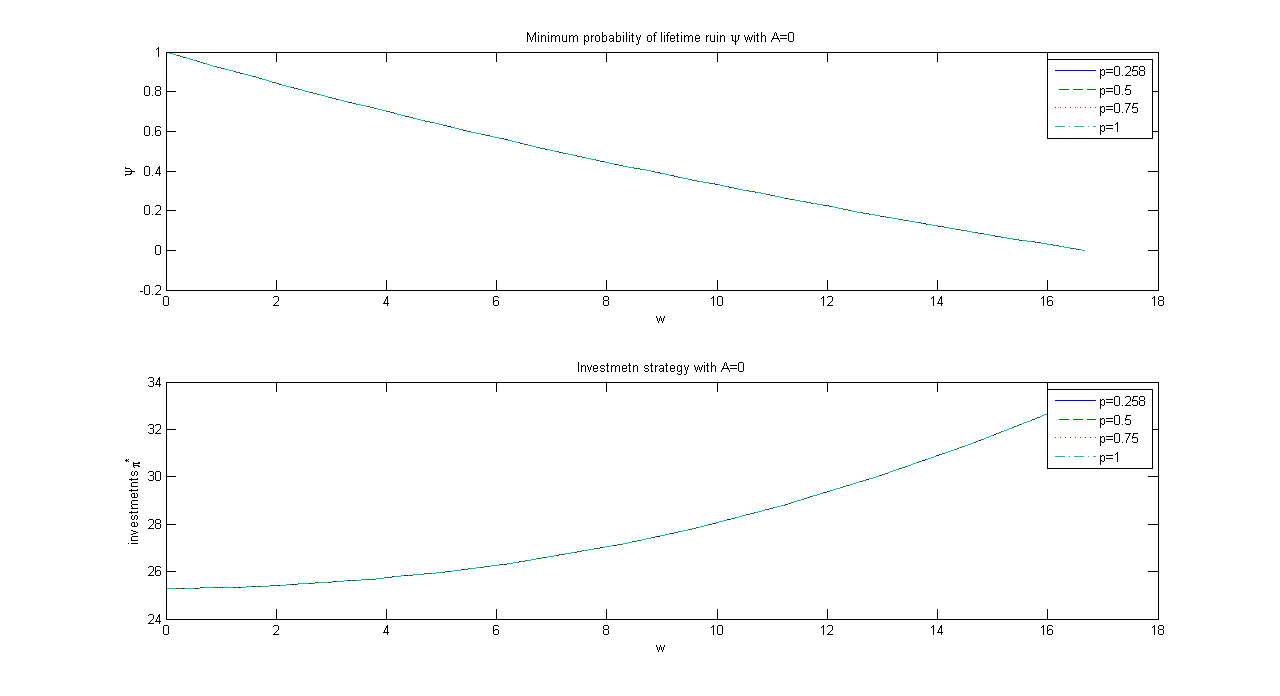}
	\caption{Ruin probabilities and optimal investment strategies for different $p$ when $A$ is $0$}
	\label{fig:2.1}
\end{figure}

\begin{figure}
	\centering
		\includegraphics[angle=90, width=1.00\textwidth, height=0.9\textheight]{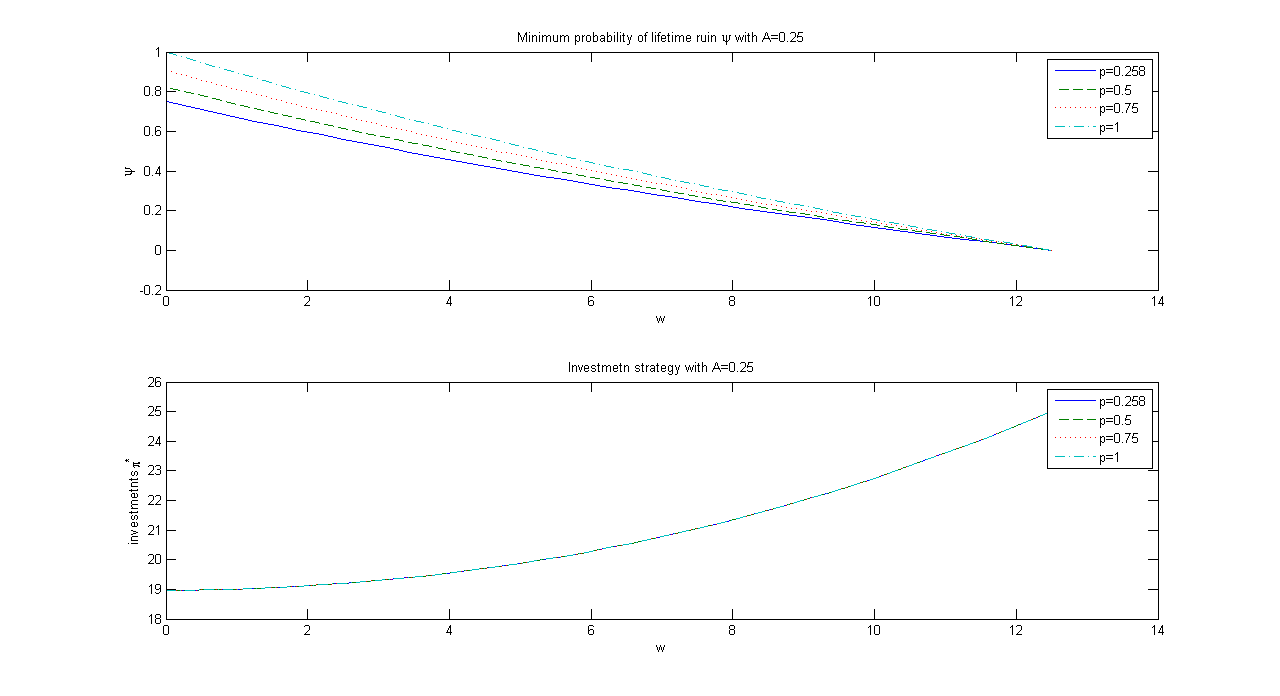}
	\caption{Ruin probabilities and optimal investment strategies for different $p$ when $A$ is $0.25$}
	\label{fig:2.2}
\end{figure}

\begin{figure}
	\centering
		\includegraphics[angle=90, width=1.00\textwidth, height=0.9\textheight]{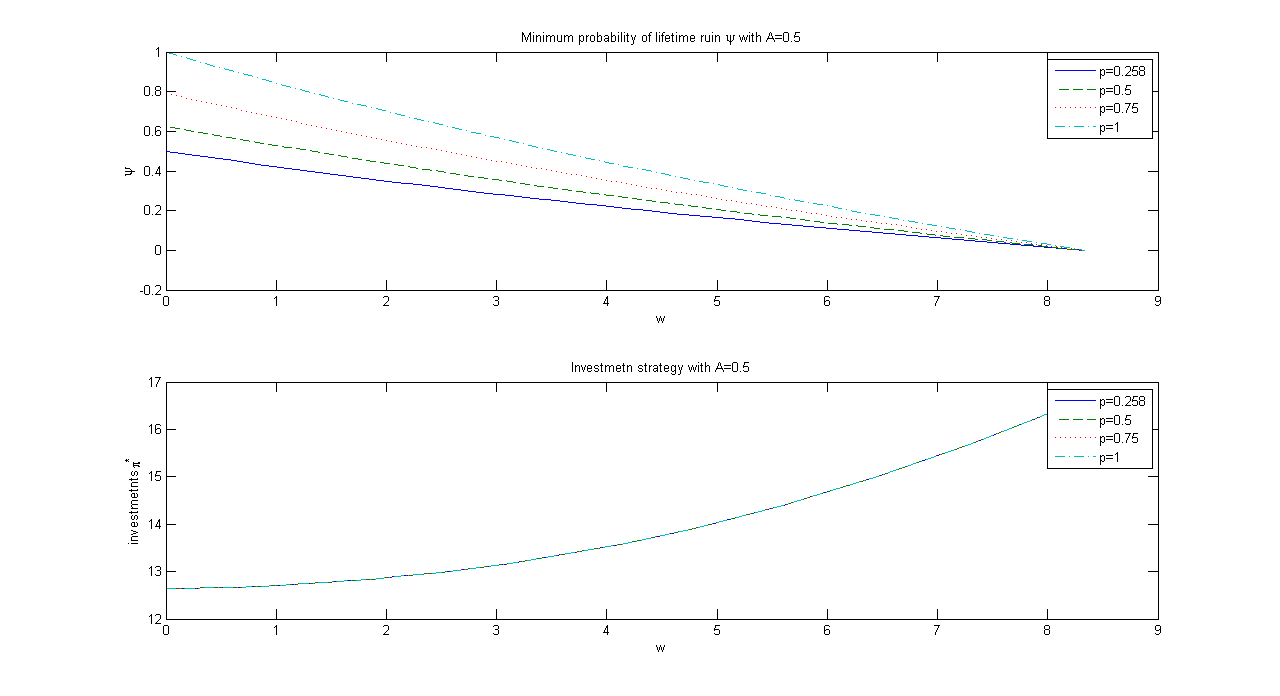}
	\caption{Ruin probabilities and optimal investment strategies for different $p$ when $A$ is $0.5$}
	\label{fig:2.3}
\end{figure}

\begin{figure}
	\centering
		\includegraphics[angle=90, width=1.00\textwidth, height=0.9\textheight]{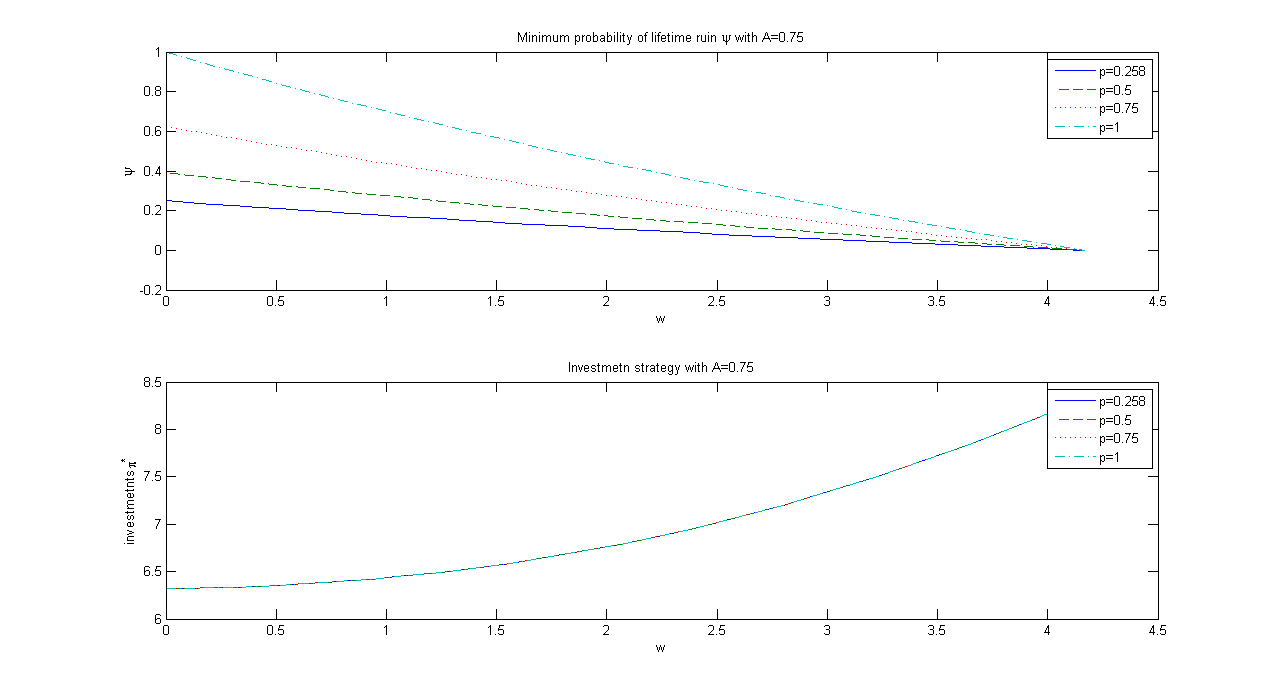}
	\caption{Ruin probabilities and optimal investment strategies for different $p$ when $A$ is $0.75$}
	\label{fig:2.4}
\end{figure}

\begin{figure}
	\centering
		\includegraphics[angle=90, width=1.00\textwidth, height=0.9\textheight]{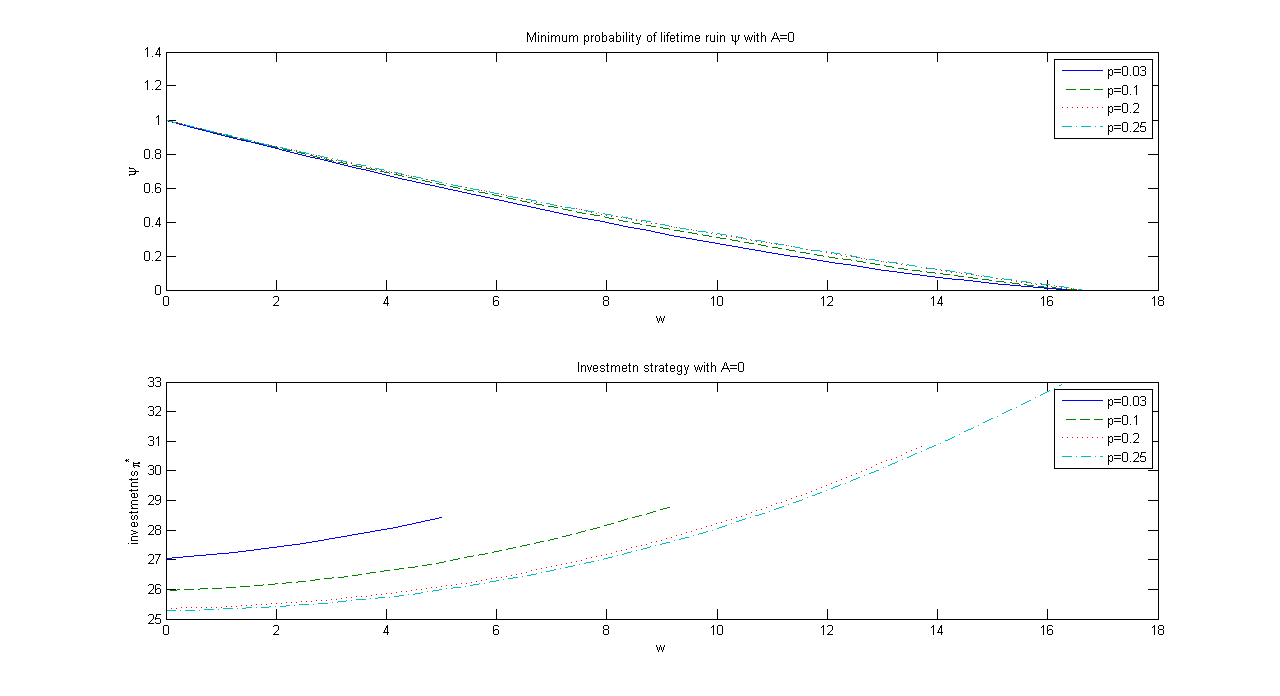}
	\caption{Ruin probabilities and optimal investment strategies for different $p$ when $A$ is $0$}
	\label{fig:3.1}
\end{figure}

\begin{figure}
	\centering
		\includegraphics[angle=90, width=1.00\textwidth, height=0.9\textheight]{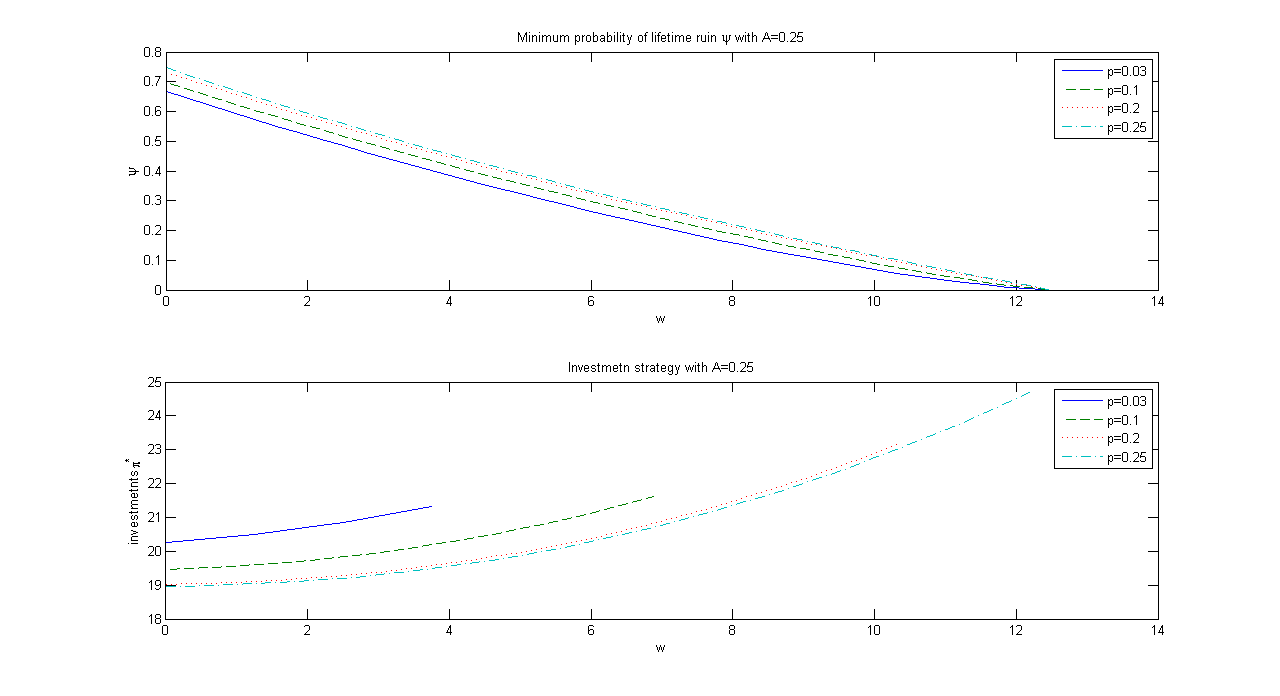}
	\caption{Ruin probabilities and optimal investment strategies for different $p$ when $A$ is $0.25$}
	\label{fig:3.2}
\end{figure}

\begin{figure}
	\centering
		\includegraphics[angle=90, width=1.00\textwidth, height=0.9\textheight]{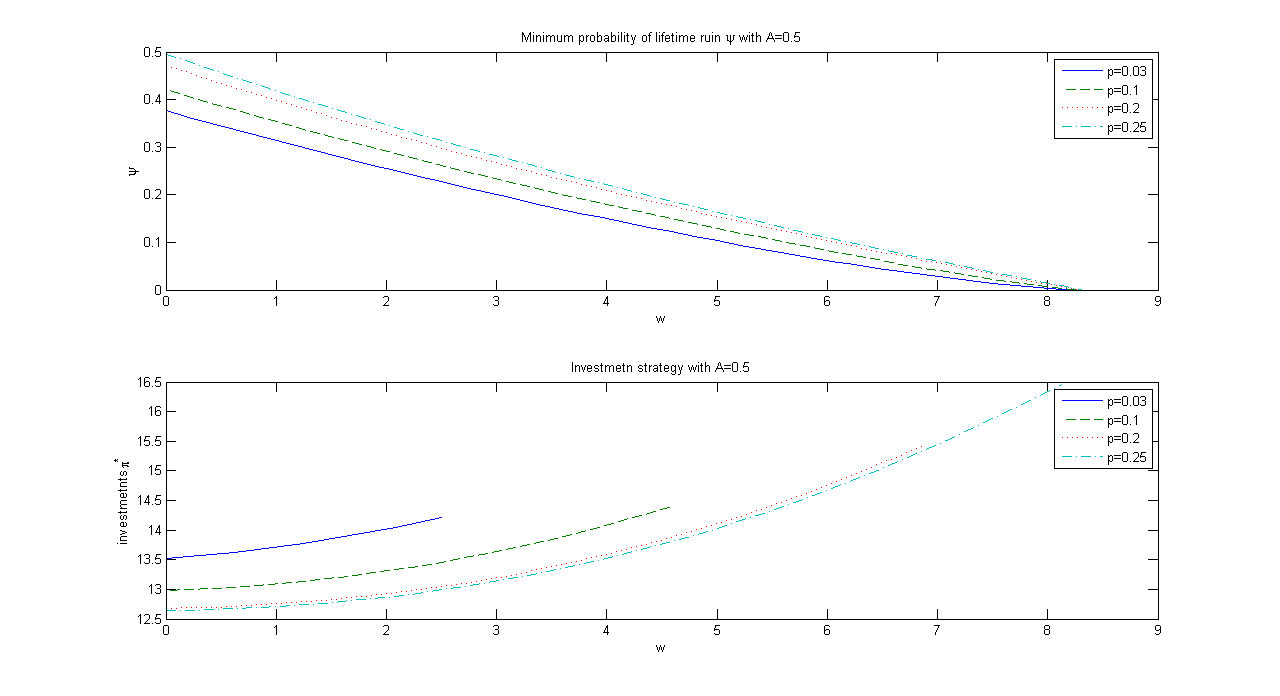}
	\caption{Ruin probabilities and optimal investment strategies for different $p$ when $A$ is $0.5$}
	\label{fig:3.3}
\end{figure}

\begin{figure}
	\centering
		\includegraphics[angle=90, width=1.00\textwidth, height=0.9\textheight]{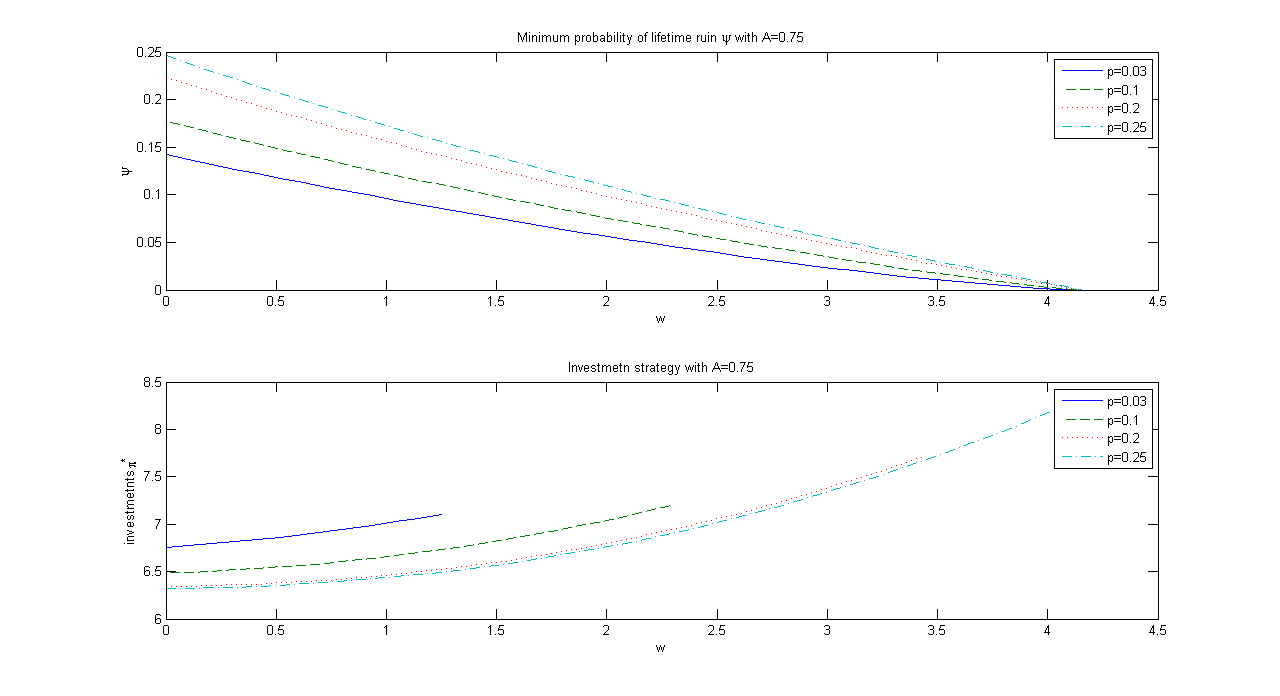}
	\caption{Ruin probabilities and optimal investment strategies for different $p$ when $A$ is $0.75$}
	\label{fig:3.4}
\end{figure}

\begin{figure}
	\centering
		\includegraphics[width=0.70\textwidth, height=0.7\textheight]{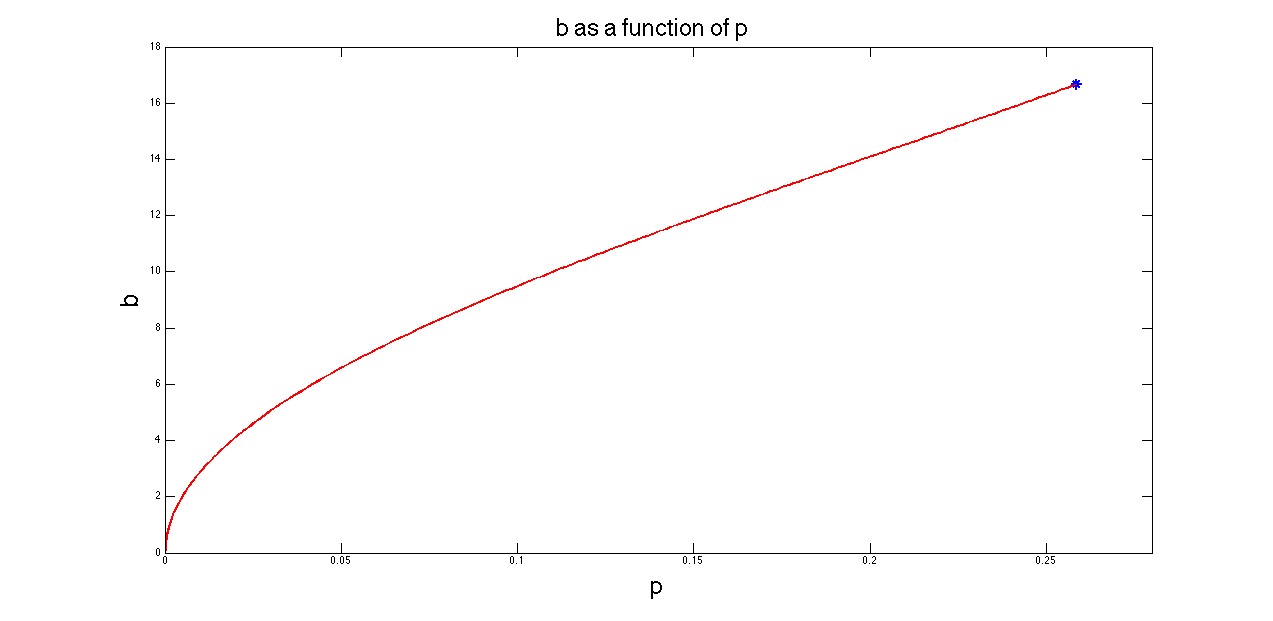}
	\caption{value of $b$ with different $p$}
	\label{fig:bp}
\end{figure}

\bibliography{references}
\bibliographystyle{abbrvnat}

\end{document}